\documentclass[12pt,a4paper]{article}
\usepackage[latin1]{inputenc}
\usepackage{amsmath}
\usepackage{amsthm}
\usepackage{amsfonts}
\usepackage{amssymb}
\usepackage{graphicx}
\usepackage{nicefrac}
\usepackage{fullpage}
\usepackage{hyperref}
\usepackage{cleveref}
\usepackage{multirow}
\usepackage{pgfplots}
\usepackage[linesnumbered,ruled]{algorithm2e}
\usepackage{caption}

\usepackage{tikz}
\usetikzlibrary{calc,positioning,shapes,backgrounds, decorations.pathreplacing}

\newtheorem{theorem}{Theorem}[section]
\newtheorem{prop}[theorem]{Proposition}
\newtheorem{claim}[theorem]{Claim}
\newtheorem{lemma}[theorem]{Lemma}
\newtheorem{remark}[theorem]{Remark}

\newtheorem{defi}[theorem]{Definition}

\newcommand{\ceil}[1]{ \left \lceil {#1} \right \rceil}
\newcommand{\Rin}{\mathcal{R_{\textup{in}}}}
\newcommand{\Rout}{\mathcal{R_{\textup {out}}}}
\newcommand{\delin}{\delta_{\textup{in}}}
\newcommand{\delout}{\delta_{\textup{out}}}
\newcommand{\epsout}{\epsilon_{\textup{out}}}
\newcommand{\epsin}{\epsilon_{\textup{in}}}
\newcommand{\Cout}{C_{\textup{out}}}
\newcommand{\Cin}{C_{\textup{in}}}
\newcommand{\ssub}{s_{\textup{sub}}}
\newcommand{\ed}{\textup{ED}}
\newcommand{\zo}{\{0,1\}}
\newcommand{\rate}{\text{Rate}}

\newcommand*\samethanks[1][\value{footnote}]{\footnotemark[#1]}

\tikzset{
	short/.style={draw,rectangle,text height=10pt,text depth=1pt,
		text width=1.3cm,align=center,fill=gray!10},
	long/.style={short,text width=3cm},
	buff/.style={short, fill=gray!30, text width=1cm}
}

\begin{document}
	\title{Explicit and Efficient Constructions of Coding Schemes for the Binary Deletion Channel and the Poisson Repeat Channel}
	
	\author{
		Roni Con\thanks{Department of Computer Science, Tel Aviv University, Tel Aviv, Israel. The research leading to these results has received funding from the Israel Science Foundation (grant number 552/16) and from the Len Blavatnik and the Blavatnik Family foundation. 
		} \and Amir Shpilka\samethanks[1] \\
	}
	\date{}
	\maketitle
	\begin{abstract}
		This work gives an explicit construction of a family of error correcting codes for the binary deletion channel and for the Poisson repeat channel. In the binary deletion channel with parameter $p$ (BDC$_p$) every bit is deleted independently with probability $p$. A lower bound of $(1-p)/9$ is known on the capacity of the BDC$_p$ \cite{mitzenmacher2006simple}, yet no explicit construction is known to achieve this rate. We give an explicit family of codes of rate $(1-p)/16$, for every $p$. This improves upon the work of Guruswami and Li \cite{guruswami2017efficiently} that gave a construction of rate $(1-p)/120$. The codes in our family have polynomial time encoding and decoding algorithms.
		
		Another channel considered in this work is the Poisson repeat channel with parameter $\lambda$ (PRC$_{\lambda}$) in which every bit is replaced with a discrete Poisson number of copies of that bit, where the number of copies has mean $\lambda$. We show that our construction works for this channel as well. As far as we know, this is the first explicit construction of an error correcting code for PRC$_{\lambda}$. 
	\end{abstract}

	\newpage
	\tableofcontents
	\newpage	
	
	\section{Introduction}

	This work deals with constructing error correcting codes for the channels called the binary deletion channel (BDC for short) and the Poisson repeat channel (PRC for short).
	
	Loosely speaking, a channel is a medium over which messages are sent. A channel is defined by the way in which it introduces errors to the transmitted messages (also called codewords when they come from an error correcting code).	
	Before describing the channels that we consider in this work, we first discuss the two main error models - a worst case model and an average case model.
	
	The first model, which is very common in the theory of computation and has found many applications there, is called the Hamming model \cite{hamming1950error}. This is a worst case setting in which a transmitted message is subjected to an adversarial corruption of a fraction $p$ of its entries and we must recover the original message regardless of the location of the errors. Thus, if the adversary is allowed to corrupt  a fraction $p$ of the entries of a transmitted message, then an error correcting code for this channel that allows perfect recovery is a subset of the messages such that any two codewords (i.e. elements of the code) have normalized hamming distance larger than $2p$.
	The second error model, which is the one relevant to our work, was first considered by Shannon in his pioneering work \cite{shannon1948mathematical}. This is an average case model in which a transmitted message is subjected to a random corruption such as bit flips, bit erasures, bit deletions, etc., where each bit is corrupted independently at random according to some distribution.\footnote{This description corresponds to a memoryless channel, which is the most common model.} A channel is basically determined by the probability distribution of corruptions. Since the corruption is random, it can be the case that the whole word is corrupted. In particular, in this setting, the most we can expect from the decoder is to decode the original word with high probability (over the randomness of the corruptions). 
	
	The two most studied channels are the Binary Erasure Channel (BEC$_p$) where each bit is independently replaced by a question mark with probability $p$ and the Binary Symmetric Channel (BSC$_p$) where each bit is independently flipped with probability $p$. 
	
	In this work we consider the BDC with parameter $p$. This channel models the situation where bits of a transmitted message are deleted (i.e. removed) from the message randomly and independently with probability $p$. In particular, if a message of length $n$ was transmitted on the BDC$_p$ then the length of the received message is concentrated around $(1-p)\cdot n$. We note that the output of the BDC is very different from that of the BEC or the BSC.  For example, if we transmit the message $1110101$ over each of the channels and corruptions occurred in locations $2$ and $5$, then the BEC will return the word $1?10?01$, the BSC will return $1010001$, and the BDC will return $11001$. In particular, while the BEC and the BSC do not affect the length of messages transmitted over them, the BDC does exactly that. Thus, unlike the BEC and BSC, the BDC causes synchronization errors. In fact, one of the main reasons for introducing the BDC was to model synchronization errors in communication.

	The motivation to study the BDC is obvious. It is not just a theoretical object as it describes a real-life scenario in which there is a loss of information that was sent on some physical layer as well as synchronization errors. Moreover, the surveys \cite{mitzenmacher2009survey,mercier2010survey} indicate that tools that were developed in the context of the BDC are useful in the study of other questions. An example of such a question is the trace reconstruction problem, which has applications in computational biology and DNA storage systems \cite{bornholt2016dna}. The problem that we study in this work is the construction of explicit error correcting codes of (relatively) high rate (we will soon explain this notion) for the BDC$_p$.

	
	
	Another model that we consider is the PRC that was first introduced in the work of Mitzenmacher and Drinea \cite{mitzenmacher2006simple}. In the PRC with parameter $\lambda$, each bit of the message is (randomly and independently) replaced with a discrete number of copies of that bit, distributed according to the Poisson distribution with parameter $0<\lambda$. In particular, with probability $e^{-\lambda}$ the bit is deleted from the message (i.e. this channel can cause synchronization errors similar to the BDC).  This channel can model, for example, messages sent using a keyboard that has tendency to get stuck so a key cannot be pressed or can get stuck and then its symbol is repeated several times. 
	While the PRC is less motivated by practical applications (we are unaware of any applications of this channel besides in the study of the BDC), it is closely related to the BDC as demonstrated in the work of Mitzenmacher and Drinea  \cite{mitzenmacher2006simple}, Drinea and Mitzenmacher \cite{drinea2007improved} and  Cheraghchi \cite{cheraghchi2018capacity}. In particular, the lower bound on the capacity (a notion that we explain shortly) of BDC$_p$ of $(1-p)/9$ \cite{mitzenmacher2006simple} relies on a reduction from the PRC$_{\lambda}$.
	We too exploit the connection between the BDC and the PRC, and using our construction for the BDC we obtain explicit constructions of error correcting codes for the PRC. 
	
	To explain the question that we study we need some basic notions from coding theory. Recall that a binary\footnote{The case of codes over non-binary alphabets is very important of course, but in this work we only focus on binary codes.} error correcting code can be described either as an encoding map $C : \zo^k \rightarrow \zo^n$ or, abusing notation, as the image of such a map $C$. The rate of such a code $C$ is $\rate(C)=k/n$, which intuitively captures the amount of information encoded in every bit of a codeword. Naturally, we would like the rate to be as large as possible, but there is a tension between the rate of the code and the amount of errors/noise it can tolerate.
	
	One of the most fundamental questions when studying a channel is to determine its capacity, i.e., the maximum achievable transmission rate over the channel that still allows recovering from the errors introduced by the channel, with high probability. Shannon proved in his seminal work \cite{shannon1948mathematical} that the capacity of the BSC$_p$ is $1-h(p)$, where $h(\cdot)$ is the binary entropy function (for $0<x<1$, $h(x)=-x\log x-(1-x)\log{1-x}$).\footnote{All logarithms in this paper are base $2$.} I.e., there are codes with block lengths going to infinity, whose rates converge to $1-h(p)$, that can recover with high probability from the errors inflicted by the channel.
	Elias \cite{	elias}, who introduced the BEC$_p$, proved that its capacity is $1-p$.

	
	What about the capacity of the BDC$_p$?  In spite of many efforts (see \cite{mitzenmacher2009survey}), the capacity of the BDC$_p$ is still not known and it is an outstanding open challenge to determine it. Yet, for the extremal cases, the asymptotic behavior is somewhat understood. In the regime where $p\rightarrow 0$ the capacity approaches to $1-h(p)$ \cite{kalai2010tight}, i.e. it approaches the capacity of the BSC$_p$. In the regime where $p\rightarrow 1$ the capacity is at least $(1-p)/9$ \cite{mitzenmacher2006simple}. This means that even if $p$ is extremely close to $1$, there are codes of positive rate that allow reliable communication over this channel. Another somewhat surprising aspect of this result is that the asymptotic behavior is only a constant off from the capacity of the related BEC$_p$. In the BEC$_p$, we know how to build codes that nearly achieve its capacity of $1-p$ for every $p$. This is not the case for the BDC$_p$, where the best explicit construction known for the regime $p \rightarrow 1$, prior to this work, has rate of $(1-p)/120$ \cite{guruswami2017efficiently}.  
	
	
	In this work we present and analyze an efficient, deterministic construction of a family of codes for the BDC$_p$ that achieves rate higher than $(1-p)/16$ for every $p$. 
	We also show that this construction yields a family of codes for PRC$_{\lambda}$ of rate $\mathcal{R} > \lambda/ 17$ for $\lambda \leq 0.5$. This further emphasizes that these channel have much in common. 
	
	\subsection{Previous Work}
	Much of the major results on the capacity of deletion type channels can be found in the excellent surveys of Mitzenmacher's and Mercier et al.   \cite{mitzenmacher2009survey,mercier2010survey}. We highlight some of the results related to the regime where $p$ tends to $1$ as this regime is the focus of this paper.
	
	The best known lower bound on the capacity is due to  Mitzenmacher and Drinea \cite{mitzenmacher2006simple} that showed a lower bound of $ (1-p)/9$ for all $p$, meaning that there are codes of this rate such that every transmitted codeword is decoded correctly with high probability. Their proof is existential and does not yield an explicit construction with this rate. As far as we know there is no explicit construction that achieves rate of $(1-p)/9$. 
	A more recent work by Guruswami and Li \cite{guruswami2017efficiently} presents a deterministic code construction for the BDC$_p$ with rate $(1-p)/120$ for all values of $p$. This rate is smaller then Mitzenmacher's bound, but it is the first construction with rate that scales proportionally to $(1-p)$ for $p\rightarrow 1$. In \cite{dalai2011new} an upper bound of $0.4143(1-p)$ was shown on the capacity of BDC$_p$ for $p\rightarrow 1$, meaning that there are no error correcting codes that achieve this rate for the BDC$_p$.\footnote{We note that Dalai's proof was computer assisted.  A recent work by Cheraghchi \cite{cheraghchi2018capacity} gave an upper bound on the capacity of the BDC$_p$ for $p\geq 1/2$ of $(1-p)\log((1+\sqrt{5})/2)$ without computer assistance.} 
	
	%
	Deletion correction is studied also in the adversarial model, i.e., when there is an adversary that can delete up to some threshold number of symbols. In fact, works dealing with the adversary model considered the more general case in which the adversary is also allowed, in addition to deletions, to insert symbols, i.e., to add a new symbol from the alphabet between two adjacent symbols in the codeword. 
	In this context of adversarial deletions and insertions, the work of Haeupler and Shahrasbi \cite{haeupler2017synchronization} gave efficient insertion-deletion (insdel for short) codes over large alphabet, which are almost optimal in rate-distance trade-off. In particular, for every $\epsilon>0$ and $\delta \in (0,1)$ there is a code $C$ with rate $1-\delta - \epsilon$ that can efficiently correct a $\delta$ fraction of insertions and deletions and its alphabet size is given by $\left| \Sigma \right| = O_{\epsilon}(1)$. We note that this construction does not give a binary code. 
	In the high rate regime, Guruswami and Wang \cite{guruswami2017deletion} showed that there are binary codes of rate $1-\tilde{O}(\sqrt{\delta})$ that can correct $\delta n$ worst-case deletions in polynomial time.
	
	\subsection{Our Results}
	In this work, we improve the construction presented in \cite{guruswami2017efficiently} and construct an explicit family of efficiently encodable and decodable codes for the binary deletion channel with rate at least $(1-p)/16$ for any $p \in (0,1)$.

	\begin{theorem} \label{BDC_theorem}
		Let $p \in (0,1)$. There exist a family of binary error correcting codes $\left \lbrace C_i \right \rbrace_{i=1}^{\infty}$ for the BDC$_p$ where the block length of $C_i$ goes to infinity as $i\rightarrow \infty$ and
		\begin{enumerate}
			\item $C_i$ can be constructed in time polynomial in its block length.
			\item $C_i$ has rate at least $(1-p)/16$.
			\item $C_i$ is decodable in quadratic time and encodable in linear time.
		\end{enumerate}
	\end{theorem}
	As mentioned earlier, we show that the same construction works for the PRC$_{\lambda}$ as well. In particular we prove,
	
	\begin{theorem} \label{PRC_theorem}
		Let $\lambda \leq 0.5$. There exist a family of binary error correcting codes $\left \lbrace C_i \right \rbrace_{i=1}^{\infty}$ for PRC$_{\lambda}$ 
		where the block length of $C_i$ goes to infinity as $i\rightarrow \infty$ and
		\begin{enumerate}
			\item $C_i$ can be constructed in time polynomial in its block length.
			\item $C_i$ has rate $\mathcal{R}_i > \lambda /17$.
			\item $C_i$ is decodable in quadratic time and encodable in linear time.
		\end{enumerate}
	\end{theorem}
	
	To the best of our knowledge, this is the first explicit construction of an error correcting code for the PRC$_{\lambda}$. 
	
	
	\subsection{Construction and Proof Overview}\label{sec:intro-const}
	
	Our construction follows the footsteps of the construction of Guruswami and Li \cite{guruswami2017efficiently} with some important modifications. We next describe the construction and then its analysis.

	\begin{figure}
		\def\nnode#1#2#3{
			\node[short, right=of #1] (#2) {#3}
		}
		
		\def\dotsnode#1#2{
			\node[long, right=of #1] (#2) {$\ldots$}
		}
		
		\def\downnode#1#2#3{
			\node[short, below=1cm of #1] (#2) {#3}
		}
		
		\begin{center}	
			\noindent\begin{tikzpicture}[node distance=-\pgflinewidth][decoration={brace}]
			
			\node[] (blind) {};
			\node[long, right=of blind] (msg) {message};	
			\downnode{blind}{out1}{$\sigma_1$};
			\nnode{out1}{out2}{$\sigma_2$};
			\dotsnode{out2}{dots};
			\nnode{dots}{outn}{$\sigma_n$};
			
			\draw[very thick,->] (msg.south west) -- (out1.north west);
			\draw[very thick,->] (msg.south east) -- (outn.north east);
			
			\node[above right=0.5cm of outn] (descon) {($1$) Outer encoding};
			
			\downnode{out1}{in1}{$c_{\sigma_1}^{(\sf in)}$};
			\downnode{out2}{in2}{$c_{\sigma_2}^{(\sf in)}$};
			\node[below right=0.13cm and 0.35cm of outn.north east, text width=4cm] (descon) {($2$) Concatenation: encoding $\sigma_i$ using inner encoding};
			\dotsnode{in2}{indots};
			\downnode{outn}{inn}{$c_{\sigma_n}^{(\sf in)}$};
			\draw[very thick,->] (out1) -- (in1);
			\draw[very thick,->] (out2) -- (in2);
			\draw[very thick,->] (outn) -- (inn);

			\def\leftdownnode#1#2#3{
				\node[short, below left=1.3cm and 0.5cm of #1] (#2) {#3}
			}
			
			\def\buffnode#1#2{
				\node[buff, right=of #1] (#2) {$0 \ldots 0$}
			}
			
			\leftdownnode{in1}{newc1}{$c_{\sigma_1}^{(\sf in)}$};
			\buffnode{newc1}{buf1};
			\nnode{buf1}{newc2}{$c_{\sigma_2}^{(\sf in)}$};
			\buffnode{newc2}{buf2};
			\dotsnode{buf2}{dotsbuf};
			\node[below right=0.5cm of inn] (descon) {($3$) Buffering};
			\buffnode{dotsbuf}{buf3};
			\nnode{buf3}{newcn}{$c_{\sigma_n}^{(\sf in)}$};
			
			\draw[very thick,->] (in1) -- (newc1);
			\draw[very thick,->] (in2) -- (newc2);
			\draw[very thick,->] (inn) -- (newcn);

			\def\onebitnode#1#2#3{
				\node[short, text width=0.2cm, right=of #1] (#2) {#3}
			}
			
			\def\twobitnode#1#2#3{
				\node[short, text width=0.4cm, right=of #1] (#2) {#3#3}
			}
			
			\def\smallrun#1#2#3#4{
				\node[short, text width=0.8cm, right=of #1] (#2) {#3#3#3#3};
				\draw (#4.south west) -- (#2.north west);
				\draw (#4.south east) -- (#2.north east)
			}

			\def\bigrun#1#2#3#4{
				\node[short, text width=2.4cm, right=of #1] (#2) {#3#3#3#3#3#3#3#3#3#3#3#3};
				\draw (#4.south west) -- (#2.north west);
				\draw (#4.south east) -- (#2.north east)
			}
			
			\node[below left=0.6cm and 0.1cm of newc2] (p3) { $\cdots$ };
			\onebitnode{p3}{b1}{1};
			\twobitnode{b1}{b2}{0};
			\onebitnode{b2}{b3}{1};
			\onebitnode{b3}{b4}{0};
			\twobitnode{b4}{b5}{1};
			\node[right=of b5] (b) { $\cdots$ };
			
			\node[below left=2cm and 0.5cm of p3] (c) {$\cdots$ };
			\smallrun{c}{rb1}{1}{b1};
			\bigrun{rb1}{rb2}{0}{b2};
			\smallrun{rb2}{rb3}{1}{b3};
			\smallrun{rb3}{rb4}{0}{b4};
			\bigrun{rb4}{rb5}{1}{b5};
			\node[right=of rb5] (d) {$\cdots$ };
			
			\node[below right=3.2cm and 0.35cm of inn] (descon) {($4$) Blow-up};
			
			\draw[decoration={brace,raise=1pt,amplitude=0.5cm},decorate,line width=1pt] (p3.north east) -- (b.north west);
			
			\end{tikzpicture}
			
		\end{center}
		\caption{The encoding process.} \label{fig:M1}
	\end{figure}
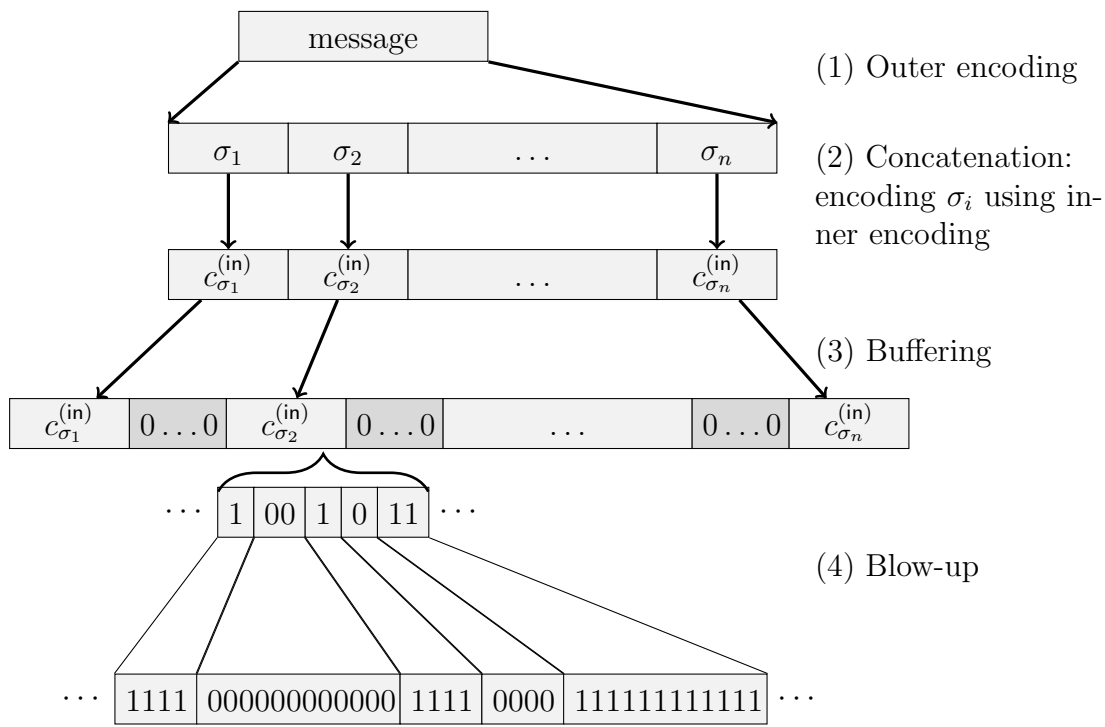

	\paragraph{Construction:}
	
	There are several layers to our construction as depicted in Figure~\ref{fig:M1} on page~\pageref{fig:M1}. The first two layers come from code concatenation while the third and fourth layers blow-up the code further by repeating symbols and inserting ``buffers'' between inner codewords. These four layers are similar to those in the construction of  \cite{guruswami2017efficiently} and the main difference between the constructions is that we use a different inner code and the blow-up in our construction is considerably smaller.
	We now describe each step in more detail.
	
	Recall that code concatenation is the operation of viewing the message as a shorter message over a larger alphabet, then applying an error correcting code over the large alphabet (the outer code) to the message and, finally, viewing each symbol of the encoded message as a short message over $\{0,1\}$, it is encoded using a binary error correcting code (the inner code).
	
	In our construction, we view the messages as strings of length $k$ over the alphabet $\Sigma=\{0,1\}^{m'}$ (where $m'$ is some constant that we later optimize). 
	As an outer code, we use the code from \cite{haeupler2017synchronization}, which is an efficient insertion-deletion code with rate close to $1$ over $\Sigma$. This code returns a word $(\sigma_1,\sigma_2, \ldots , \sigma_n) \in \Sigma^n$.
	
	
	We construct our inner code using a greedy algorithm. First we consider all binary strings of length $m$ which consist of exactly $\beta_1 m$ $1$-runs and $0.5(1-\beta_1)m$ $2$-runs (i.e. alternating blocks of $0$'s and $1$'s where each block length is $\leq 2$) where $\beta_1$ is a parameter that we will optimize later. Then, we add a codeword to our codebook if it does not contain a subsequence of length $\geq m - \delta m$ that is also a subsequence of any codeword that is already in our codebook. Note that even though the construction time is exponential in $m$, as $m=O(1)$ in our construction this does not affect the run time by more than a constant factor.
	
	The encoding process is thus as follows (see also figure \ref{fig:M1}). We first encode the message using the outer code to a codeword of length $n$ over $\Sigma$. Then, the concatenation process takes every symbol, $\sigma_i$ of the outer codeword and maps it to a codeword from the inner code, i.e., a concatenated codeword is of the form $c_1 \circ c_2 \circ \cdots \circ c_n$ where $c_i = \textup{ENC}_{\textup{in}}(\sigma_i)$, where $\textup{ENC}_{\textup{in}}$ is the encoding function of the inner code. 
	This is not the end of the story. In order for the concatenated code to overcome a large amount of deletions caused by the channel we add an additional layer of encoding:
	\begin{enumerate}
		\item We place long buffers of zeros (of length $B$) between inner codewords. This step helps the decoder identify where an inner codeword starts and where it ends.
		\item We replace each $1$-run with an $N_1$-run and each $2$-run will become an $N_2$-run (runs of length $N_1$ and $N_2$). This helps the decoder identify if the run in the inner codeword was a run of length $1$ or $2$.
	\end{enumerate}
	This step is also similar to the construction of \cite{guruswami2017efficiently}, however, perhaps surprisingly, since we restrict our inner codewords to have a fixed number of $1$-runs and $2$-runs, this enables us to have $N_1$ and $N_2$ considerably smaller than the blow-up parameter used in \cite{guruswami2017efficiently}. It is clear that the code construction is efficient as the outer code of \cite{haeupler2017synchronization} can be encoded efficiently and the inner code is of constant length and thus can also be encoded efficiently. The last step is clearly efficient (as $B,N_1$ and $N_2$ are constants). \\
	
	\paragraph{Decoding:}
	
	We now describe our decoding algorithm. First, we identify the buffers in order to divide the string into ``decoding windows'' that should ideally represent corrupted inner codewords. 
	Second, every decoding window is decoded in the following way:  Every run longer than some threshold $T$ is replaced with a $2$-run (of the same symbol) and every run of length $\leq T$ is replaced with a $1$-run.  
	The third step of the decoding is to use a brute force decoding algorithm on each decoded window to find the closest inner codeword. Since the inner code's block length is constant this step takes constant time for every such window and hence runs in linear time in the length of the word.
	The last step in the decoding algorithm is to run the decoding algorithm of the outer code as given in \cite{haeupler2017synchronization}. This algorithm runs in time quadratic in the outer code's block length. Hence the total run time of our decoder is quadratic in the length of the message.
	
	\paragraph{Analysis:}
	
	Our analysis classifies errors to three types: 
	\begin{enumerate}
		\item Buffer deletions: these are deletions that caused a buffer between inner codewords to completely disappear.
		\item Spurious buffers: these are deletions of many $1$'s that caused the algorithm to mistakenly identify a buffer inside an inner codeword.
		\item Wrong decoding of inner codewords: these occur when the algorithm fails to decode correctly a corrupted inner codeword. \end{enumerate}
	The first and second error types can happen in the first stage of the algorithm, i.e., when the decoder identifies the buffers between blown-up inner codewords. First, the decoder might not identify a buffer when a large portion of the buffer was deleted and second, the decoder might mistakenly think that there is a buffer inside an inner codeword if many consecutive runs of the symbol $1$ were deleted. We show by using simple concentration bounds that both error types happen with exponentially small probability in $m$, the inner code block length (as $m=O(1)$ this is a constant probability, but it is still small enough to allow our construction to work). 
	The third error type we consider is when the edit distance between the sent inner codeword and the corresponding string obtained from the second step of the decoding algorithm is greater than $\delta_{\textup{in}}m$, the inner code's decoding radius. In this case, the decoding algorithm of the inner code might output a wrong codeword. 
	While this can happen, we show that the \emph{expected} edit distance between the original inner codeword and the decoded inner codeword\footnote{In the proof we use the term decoded window as we are never really sure when a codeword started and ended, but this does not affect the intuition.} 
	is smaller than $\delta_{\textup{in}} m$, for a large enough $m$, and furthermore, the edit distance is concentrated around its mean. Hence, we expect to decode successfully most of the inner codewords. 
	Finally, we show that this reasoning implies that the decoding algorithm of the outer code, which is executed at the last step of our decoding algorithm, succeeds with probability $1-\exp(-\Omega(n))$.
	
	In terms of complexity, we show that even though the construction and decoding of the inner code are exponential in the inner code's block length, the overall complexity (construction, encoding, and decoding) is dominated by the complexity of the outer code which has efficient encoding and decoding algorithms thanks to  \cite{haeupler2017synchronization}.
	
	\paragraph{Comparison to \cite{guruswami2017efficiently}.} 
	We end this high-level summary by elaborating more on the main similarities and differences between our construction and the construction of Guruswami and Li \cite{guruswami2017efficiently}. Our scheme, as well as our decoding algorithm, follow closely the scheme and algorithm of Guruswami and Li. In particular, at a high level, the encoding layers are the same as in \cite{guruswami2017efficiently}, meaning that both constructions use concatenation with the outer code from \cite{haeupler2017synchronization}, place long buffers between inner codewords and blow-up the code. Since the encoding layers are similar, the decoding steps in both papers are also similar: first identify the buffers, then use a threshold to distinguish between $1$-runs and $2$-runs, then use brute force to decode the inner codewords, and finally use the decoder of \cite{haeupler2017synchronization}. 
	The main differences between our scheme and the scheme from \cite{guruswami2017efficiently} are in the inner code that is used, the blow-up process which is finer in our scheme and our analysis which is more fine-tuned:
	
	\begin{itemize}
		\item The inner code that was used in \cite{guruswami2017efficiently}  has the property that every codeword consists of $1$-runs and $2$-runs, but they do not have restriction on the number of $1$-runs and $2$-runs. In contrast, in our work, all inner codewords have the same number of $1$-runs and $2$-runs. This property allows us to increase the rate of the inner code compared to \cite{guruswami2017efficiently} (See Propositions~\ref{deletion_prop} and \ref{inner_insertion_prop} and the discussion following them), while maintaining its robustness against insertions and deletions.

		\item In \cite{guruswami2017efficiently} the authors blow-up the code by replacing every bit with $60/(1-p)$ copies of that bit. 
		Instead of blowing-up every single bit, we blow-up each $1$-run to an $N_1$-run and each $2$-run to an $N_2$-run where $N_1\neq N_2$ and both are significantly smaller than $60/(1-p)$ ($N_1 \approx 6/(1-p)$ for example). Thus, the effect of the blow-up on the rate of our code is significantly smaller than in \cite{guruswami2017efficiently}.
		
		
		\item We improve on the analysis in \cite{guruswami2017efficiently}, of the edit distance between decoded inner codewords and the original inner codewords, by better accounting the effect of decoding errors on the edit distance. 
		One more improvement lies in our analysis where instead of using the Chernoff bound to upper bound the probability of certain events, we use the fact that binomial distributions with fixed expectations converge to a Poisson distribution. This gives a better upper bound which eventually leads to some saving when optimizing parameters. Our analysis further highlights the tight connection between the BDC and the PRC via the convergence of the binomial distribution to the Poisson distribution.
		
	\end{itemize}
	
	These modifications, as well as a careful choice of parameters, is the reason for the great saving in the rate compared to  \cite{guruswami2017efficiently}.

	\subsection{Organization}
	
	The paper is organized as follows. In Section~\ref{sec:prelim} we introduce the basic notation as well as some well known facts from probability and from previous papers. Section~\ref{sec:inner} contains the construction of our inner code. In Section~\ref{sec:construction} we give our construction and in Section~\ref{sec:analyze} we give its analysis. We give slightly improved bounds for fixed values of $p$ in Section~\ref{sec:fixed}.
	Finally, Section~\ref{sec:poisson} explains how to carry our construction and analysis to the PRC.
	
	\section{Preliminaries}\label{sec:prelim}
	For an integer $k$, we denote $[k]=\{1,2,\ldots,k\}$. 
	Throughout this paper, $\log(x)$ refers to the base-$2$ logarithm and $h(x)$ denotes the binary entropy function, that is,
	$h(x) = -x\log(x) - (1-x)\log(1-x)$, for $0<x<1$. We use $\Sigma$ to denote an alphabet and $\Sigma^{*}$ to denote all the finite length strings over $\Sigma$. For $s\in \Sigma^*$ we denote by $\left| s \right|$ the length of $s$. 
	\begin{defi}
		Let $s \in \Sigma^{*}$. The operation in which we remove a symbol from $s$ is called a \emph{deletion} and the operation in which we place a new symbol from $\Sigma$ between two consecutive symbols in $s$ is called an \emph{insertion}. 
		
		A \emph{substring} of $s$ is a string obtained by taking consecutive symbols from $s$.
		A \emph{subsequence} of $s$ is a string obtained by removing some (possibly none) of the symbols in $s$. 
	\end{defi}
	\begin{defi}
		
		
		Let $s,s'\in \Sigma^{*}$. 
		A \emph{longest common subsequences} between $s$ and $s'$, is a subsequence $s_\textup{sub}$ of both $s$ and $s'$, of maximal length. We denote by $\left| \textup{LCS}(s,s') \right|$ the length of a longest common subsequence.\footnote{Note that a longest common subsequence may not be unique as there can be a number of subsequences of maximal length.} 
		
		The \emph{edit distance} between $s$ and $s'$, denoted by ED$(s,s')$, is the minimal number of insertions and deletions needed in order to turn $s$ into $s'$.
	\end{defi}
	
	\begin{lemma}[See e.g.  Lemma 12.1 in \cite{crochemore2003jewels}]\label{lem:lcs}
		It holds that $\textup{ED}(s,s') = \left|s\right| + \left|s' \right| - 2 \left| \textup{LCS}(s,s') \right|$.
	\end{lemma}
	
	
	\begin{defi}
		Let $s$ be a string. A \emph{run} $r$ in $s$ is a single-symbol substring of $s$ such that the symbol before the run and the symbol after the run are different from the symbol of the run. A run of length $\ell$ will be denoted as $\ell$-run. 
	\end{defi}
	
	For example, consider the string $\langle 0111001\rangle$. It can be written as the (string) concatenation of the alternating runs $ 0 \circ 111 \circ 00 \circ 1$. Clearly, every binary string is a concatenation of runs of alternating symbols. The following lemma of Levenshtein will be useful in the analysis of the rate of our inner code.
	
	\begin{lemma} \cite{levenshtein1966binary} \label{Lev_lemma}
		Let $s$ be a string and let $r(s)$ be the number of runs in $s$. There are at most  
		\[
		\binom{r(s)+d-1}{d}
		\]
		different subsequences of $s$ of length $\left| s \right| - d$.
	\end{lemma}
	
	\subsection{Facts from Probability}
	We use two probability distributions in this paper. The \emph{binomial distribution} with parameters $n$ and $p$, denoted  $\textup{Bin} (n,p)$, is the discrete probability distribution of the number of successes in a sequence of $n$ independent trials, where the probability of success in each trial is $p$ and the probability of failure is $1-p$.
	The second distribution is the \emph{discrete Poisson distribution} with parameter $\lambda$, denoted as $\text{Poisson}(\lambda)$ which is defined with the following probability mass function
	\[
	\Pr\left[ X = k \right] = \frac{e^{-\lambda} \lambda^k}{k!}.
	\]
	A well known fact about Poisson distribution is 
	\begin{lemma} \cite[Lemma 5.2]{mitzenmacher2005probability} \label{lem:poisson-sum}
		Let $X$ and $Y$ be two independent Poisson random variables with parameters $\mu_1$ and $\mu_2$. I.e., $X\sim \textup{Poisson}(\mu_1)$ and $Y\sim \textup{Poisson}(\mu_2)$. Then $Z = X+Y$ is a Poisson random variable with parameter $\mu_1 + \mu_2$. 
	\end{lemma}	
	We shall use the following simple lemma in our analysis:
	\begin{lemma} \label{lem:poisson_mono}
		Fix $T$ to be a non negative integer and let $Y(\lambda) \sim \textup{Poisson}(\lambda)$. Then the function 
		\[
		f(\lambda) := \Pr [Y(\lambda) \leq T] = e^{-\lambda} \sum_{i=0}^{T} \frac{(\lambda)^i}{i!}  
		\]
		is monotonically decreasing in $\lambda$.
	\end{lemma} 
	\begin{proof}
		It holds that 
		\begin{align*}
		\frac{\text{d} f}{\text{d} \lambda}(\lambda) = -e^{-\lambda} \sum_{i=0}^{T}\frac{(\lambda)^i}{i!} + e^{-\lambda} \sum_{i=0}^{T-1} \frac{(\lambda)^i}{i!} = -e^{-\lambda} \frac{(\lambda)^T}{T!} < 0 \;.
		\end{align*}
	\end{proof}
	
	The next theorem shows that if we let $n$ tend to infinity and $p$ tend to zero under the restriction that $p\cdot  n=\lambda$, then the binomial distribution converges to the Poisson distribution with parameter $\lambda$.
	\begin{theorem} \cite[Theorem 5.5]{mitzenmacher2005probability} \label{bin_p_mitz} 
		Let $\lambda>0$ be fixed. Let $\{X_n\}$ be a sequence of binomial random variables such that $X_n \sim B(n,p)$, and $\lim_{n\rightarrow \infty}np = \lambda$. Then, for any fixed $k$,
		\[
		\lim_{n\rightarrow \infty} \Pr[X_n = k] = \frac{e^{-\lambda} \lambda ^k}{k!} \;.
		\]
	\end{theorem}
	
	The next theorem provides more information about the binomial distribution in the regime where $n p = \lambda$. Specifically, it tells us when is $\Pr[X\leq T]$ an increasing function of $n$.
	\begin{theorem} \cite{anderson1965some} \label{cum_anderson}
		Let $\{X_n\}$ be a sequence of binomial random variables with parameters $n$ and $p=\lambda/ n$. Let $T$ be some parameter. Set  $f(n):= \Pr[X_n \leq T]$.  
		\begin{enumerate}
			\item If $T \leq \lambda - 1$ then for every $n  \geq \lambda$, $f(n)$ is monotonically increasing in $n$.
			\item If $\lambda \leq T$ then for every $n \geq T$,  $f(n)$ is monotonically decreasing in $n$.
		\end{enumerate}
	\end{theorem}
	
	For concentration bounds, we will use the following versions of the Chernoff bounds. 
	
	\begin{lemma} \cite [Theorems 4.4 and 4.5]{mitzenmacher2005probability}\label{lem:chernoff}
		Suppose $X_1,\ldots,X_n$ are independent identically distributed random variables taking values in $\lbrace 0, 1\rbrace$. Let $X = \sum_{i=1}^{n}X_i$ and $\mu = \mathbb{E}\left[ X_i \right]$. 
		Then, for any $0 < \alpha < 1$:
		\[\Pr\left[X > (1+\alpha)n\mu \right] < e^{-\frac{\mu n\alpha ^2}{3}}\]
		and
		\[\Pr\left[X < (1-\alpha)n\mu \right] <e^{-\frac{\mu n\alpha ^2}{2}} \;.\]
	\end{lemma}
	When we have a Poisson random variable we shall use the following Chernoff bound
	
	\begin{lemma} \cite [Theorem 5.4]{mitzenmacher2005probability} \label{lem:chernoff-poisson}
		Let $X$ be a Poisson random variable with parameter $\mu$.
		\begin{enumerate}
			\item If $x> \mu$, then
			\[
			\Pr (X \geq x) \leq \frac{e^{-\mu} (e\mu)^{x}}{x^x} \;.
			\]
			\item If $x < \mu$,
			\[
			\Pr (X \leq x) \leq \frac{e^{-\mu} (e\mu)^{x}}{x^x} \;.
			\]
		\end{enumerate}
	\end{lemma}

	Another concentration bound we use is Hoeffding's inequality
	\begin{theorem} \cite[Theorem 2]{hoeffding1994probability} \label{thm:hoeffding}
		If $X_1, X_2, \ldots, X_n$ are independent random variables with finite first and second moment and $a_i \leq X_i \leq b_i$ for $1\leq i \leq n$. Let $X = \sum_{i=1}^{n}X_n$ and $\mu = \mathbb{E}[X]$ then for $t>0$
		\[
		\Pr[X - \mu > t] < \exp\left(-\frac{2t^2}{\sum_{i=1}^{n}(b_i - a_i)^2}\right) \;.
		\] 
	\end{theorem}
	To approximate binomial coefficients we shall use the following lemma
	\begin{lemma} \label{lemma:stirling}
		For any $n, k \in \mathbb{N}$ such that $k/n \leq 1/2$ we have,
		\[
		2^{nh\left(\frac{k}{n}\right) - O(\log n)} \leq \binom{n}{k} \leq 2^{n h\left(\frac{k}{n}\right)} \;.
		\]
	\end{lemma}
	The proofs of the bounds follow from Stirling's formula, e.g., see \cite[Section 3.2]{guruswami2012essential}

	\subsection{The Code of Haeupler and  Shahrasbi \cite{haeupler2017synchronization}}\label{sec:HS}

	Our construction relies on the following code of Haeupler and  Shahrasbi \cite{haeupler2017synchronization}.
	
	
	\begin{theorem} [{\cite[Theorem 1.1]{haeupler2017synchronization}}] \label{Haupler_thm}
		For every $\epsout >0$ and $\delout\in (0,1)$ there exists $n_0$ so that for every $n>n_0$ there is an integer $k$ satisfying $ k/n > 1 - \delout -\epsout$, an alphabet $\Sigma$ of size $ O_{\epsout}(1)$ and
		an encoding and decoding maps $E:\Sigma^k \mapsto \Sigma^n$, $D :\Sigma^* \mapsto \Sigma^k$, respectively,  such that if $\ed(E(x),y) \leq \delout n$ then $D(y) = x$. Further
		$E$ and $D$ are explicit and can be computed in linear and quadratic 
		time in $n$, respectively.
	\end{theorem}
	
	We shall denote with $\Rout := k/n$ the rate of this code, which will be used as the outer code in our construction.

	\section{The Inner Code}\label{sec:inner}
	
	In this section we describe the construction of our inner code. Before giving the construction we define a set of strings from which we shall pick our codewords.
	
	\begin{defi}\label{def:S}
		We denote with $S \subset \{0,1\}^*$  the set containing all binary strings $s$ that start and end with the symbol $1$ and that contain only $1$-runs and $2$-runs. 
		
		Let $\beta_1 \in [0,1]$. Define $S_{m,\beta_1}\subset S$ to be the set of all $s\in S$ of length $m$, such that the number of $1$-runs in $s$ is exactly $\beta_1 m$ and the number of $2$-runs in $s$ is exactly $\beta_2m = (1-\beta_1)m/2 $. Denote $\beta:= \beta_1 + \beta_2$. 
	\end{defi}
	
	\begin{remark}\label{rem:odd}
		We observe that as every string in $S$ begins and ends with the same symbol, the number of runs in it is odd. 
		Similarly,  $\beta m$, in the definition of  $S_{m,\beta_1}$,  is an odd integer.
	\end{remark}
	
	
	Our goal in this section is to construct a code $C \subset S_{m,\beta_1}$ 
	such that the length of a longest common subsequence of any two different codewords is $< m -\delta m$.
	
	We construct this code using the natural greedy algorithm:  We consider all strings in $S_{m,\beta_1}\subset S$ and greedily choose strings that are far from each other. To reason about the parameters of the code, we need the following propositions. 
	
	The first proposition gives an upper bound on the size of the ``deletion ball'', i.e., given a string $s\in S_{m, \beta_1}$ it upper bounds the number of different subsequences of $s$ of length $m - \delta m $, that belong to $S$.

	\begin{prop} \label{deletion_prop}
		Let $s\in S_{m,\beta_1}$. It holds that 
		\[
		\#\left\{ s'\in S \mid s' \text{ is a subsequence of } s \text{ and } |s'| =m-\delta m \right\} \leq  \binom{(\beta +\delta)m }{\delta m} \;.
		\]
	\end{prop}
	
	\begin{proof}
		
		By definition, $s$ is a binary string that contains exactly $\beta m$ runs. Let $H$ be the set of all the subsequences obtained from $s$ by applying $\delta m$ deletions. According to Lemma~\ref{Lev_lemma}, the size of $H$ is at most $ \binom{\beta m + \delta m-1}{\delta m}< \binom{\beta m + \delta m}{\delta m}$. Clearly if we restrict further and consider only those strings in $H\cap S$ we can only decrease the size of the set.
	\end{proof}
	
	The second proposition upper bounds the size of the ``insertion ball'', i.e., given a string $s'\in S$ of length $m-\delta m$, it gives an upper bound on the number of strings $s\in S_{m,\beta_1}$ that can be obtained from $s'$ by performing $\delta m$ insertions. The proof of this proposition is considerably more elaborated.

	\begin{prop} \label{inner_insertion_prop}
		Fix $\ssub \in S$ such that $\left| \ssub \right| = m - \delta m$. The number of binary strings in $S_{m,\beta_1}$ that contain $\ssub$ as a subsequence is at most
		\[
		O(\delta m)  \cdot \binom{\beta_1 m + \beta_2 m}{\delta m} \;.
		\]
	\end{prop}
	
	We note that the equivalent propositions from \cite{guruswami2017efficiently} gave upper bounds of $\binom{m}{\delta m}$ and $O(\delta m)  \cdot \binom{m}{\delta m}$ respectively. The main reason for our saving is that we restrict our codewords to have exactly $\beta_1 m$ $1$-runs and $\beta_2 m$ $2$-runs. This saving is one of the places where we improve upon \cite{guruswami2017efficiently}. This improvement affects the rate of the inner code as we shall later see.
	
	The proof of the proposition relies on an algorithm for generating all strings $s \in S_{m,\beta_1}$ such that $\ssub$ is a subsequence of $s$. As in \cite[Lemma 2.3]{guruswami2017deletion}, in order to avoid over counting, we will generate all such $s$ by finding the lexicography first occurrence of $\ssub$ in $s$.  We first explain the idea behind the algorithm and then prove Proposition~\ref{inner_insertion_prop}.
	
	Denote  $\ssub= \langle b_1 b_2 \ldots  b_{m-\delta m} \rangle$, where $b_i\in\{0,1\}$. In order to obtain a string $s\in S_{m,\beta_1}$ from $\ssub$  we need to choose indices $1 \leq n_1 < n_2 < \ldots < n_{m-\delta m} \leq m$ for the bits of $\ssub$ in $s$. Moreover, to make sure that the locations chosen are indeed the lexicography first occurrence of $\ssub$ in $s$, the entries between $n_i$ and $n_{i+1}$ (for $1\leq i \leq m-\delta m -1$) must contain the opposite bit of the symbol in location $n_{i+1}$.
	
	Since both $\ssub$ and $s$ consist of just $1$-runs and $2$-runs, this puts some restrictions on the embedding of $\ssub$ in $s$, e.g., we cannot have $n_{i+1} - n_i \geq 3$ (all locations between them (and maybe longer) are identical and hence give a too long run). In particular, and more formally, we have the following restrictions
	\begin{enumerate}
		\item The first bit in $\ssub$ must be located as the first bit in $s$. This is because the first bit in $s$ must be a $1$ bit. I.e., $n_1 = 1$.
		\item Let $b_i$ be a $1$-run in $\ssub$ and assume w.l.o.g. that it is a $0$ bit. Its location, $n_i$, must be chosen such that the location of the next bit in $\ssub$, $b_{i+1}$, is either 
		\begin{enumerate}
			\item $n_{i+1} = n_i + 1$. I.e., $\langle b_i,b_{i+1}\rangle =\langle 01 \rangle$ in $\ssub$ is mapped to $\langle 01 \rangle$ in $s$, or
			\item \label{first_type_ins} $n_{i+1} = n_i + 2$. I.e., $\langle 01 \rangle$ in $\ssub$ is mapped to $\langle 001 \rangle$ in $s$.
		\end{enumerate}
		The case where $b_i= 1 $ is completely analogous. 
		\item Let $b_i, b_{i+1}$ be a $2$-run in $\ssub$ (i.e., the symbols of $b_i$ and $b_{i+1}$ are the same) and assume w.l.o.g. that both symbols are $0$. The locations $n_i, n_{i+1}, n_{i+2}$ of $b_i, b_{i+1}, b_{i+2}$ ($b_{i+2}$ is a $1$ bit) in $s$ must be chosen in accordance with one of the following cases:
		\begin{enumerate}
			\item $n_{i+1} = n_{i}+1$ and $n_{i+2} = n_i + 2$. I.e., $\langle 001 \rangle$ in $\ssub$ is mapped  to $\langle 001 \rangle$ is $s$.
			\item \label{second_type_ins_0}$n_{i+1} = n_{i}+2$ and $n_{i+2} = n_i + 3 $. I.e., $\langle 001 \rangle$ in $\ssub$ is mapped  to $\langle 0101 \rangle$ in $s$.
			\item \label{second_type_ins_1} $n_{i+1} = n_{i}+2$ and $n_{i+2} = n_i + 4 $. I.e., $\langle 001 \rangle$ in $\ssub$ is mapped  to $\langle 01001 \rangle$ in $s$. 
			\item  \label{second_type_ins_3} $n_{i+1} = n_{i}+3$ and $n_{i+2} = n_i + 4 $. I.e., $\langle 001 \rangle$ in $\ssub$ maps to $\langle 01101 \rangle$ in $s$.
			\item \label{second_type_ins_4}$n_{i+1} = n_{i}+3$ and $n_{i+2} = n_i + 5 $. I.e., $\langle 001 \rangle$ in $\ssub$ is mapped  to $\langle 011001 \rangle$ in $s$.
		\end{enumerate}
		The case where $b_i=1$ is completely analogous. 
		\item If $n_{m-\delta m} < m$, then the remaining bits of $s$ must be filled with $1$-runs and $2$-runs such that the total number of $1$-runs and $2$-runs is exactly $\beta_1 m$ and $\beta_2 m $, respectively.
	\end{enumerate}

	It is not hard to verify that any arrangement that does not follow the restrictions above will either contain a run of length $3$ or more, will not have the right number of $1$-runs, or will not correspond to the lexicographically first embedding of $\ssub$ in $s$.
	
	We shall think of the cases above as describing operations that can be performed on a string $s'$. E.g. if $s'= \langle 101001\rangle$ and we apply \ref{second_type_ins_4} to the last three bits in $s'$ then we will get the string $\langle  1010{\color{blue}11}0{\color{blue}0}1\rangle$, where the blue symbols are the symbols that were added from the application of \ref{second_type_ins_4} (in other words, the symbols colored black are the embedding of the original string). If we then apply, say,  \ref{first_type_ins} to the second and third bits of the new string then we will get the string 
	$\langle  10{\color{blue}0}10{\color{blue}11}0{\color{blue}0}1\rangle$
	etc.
	
	To simplify matters note that if we consider a $2$-run in $s'$, say  $\langle 001\rangle$ and we wish to apply \ref{second_type_ins_1} on it, i.e. map it to $\langle 01001\rangle$ in $s$, then we can think about this as first applying \ref{second_type_ins_0} to $\langle 001\rangle$, obtaining the string $\langle 0101\rangle$ and then applying to the last two bits \ref{first_type_ins}, getting the string $\langle 01001\rangle$. I.e. we can simulate \ref{second_type_ins_1} by first applying \ref{second_type_ins_0} and then applying \ref{first_type_ins}. 
	Similarly, we can simulate each of the operations  \ref{second_type_ins_3}  and  \ref{second_type_ins_4} using   \ref{second_type_ins_0} and then applying \ref{first_type_ins}  to the appropriate bits (for \ref{second_type_ins_4} we need to apply \ref{second_type_ins_0} and then \ref{first_type_ins}  to two different locations).

	
	Using the above terminology, we next describe an algorithm that given a string $\ssub$ generates $s\in S_{m,\beta_1}$ such that $\ssub$ is a subsequence of $s$. The algorithm will first select a subset of the $2$-runs in $\ssub$ and apply  \ref{second_type_ins_0} to them. Then it will add more $1$-runs to the resulting string, locating them to the right of the last bit. Finally, it will apply  \ref{first_type_ins} to several $1$-runs. 
	
	There is a delicate point that we wish to stress before giving the algorithm. In this last step we restrict the $1$-runs to which we can apply \ref{first_type_ins}. To illustrate why the restriction is needed, consider the following example: 
	Consider the string $\langle 001 \rangle$ and  apply \ref{second_type_ins_0} to it. This generates the string  $ \langle 0{\color{blue} 1}01\rangle $, where, as before,  the blue symbols represent the symbols that were added in the embedding. If we now apply \ref{first_type_ins} to the first two bits then we would get  $ \langle 0{\color{blue} 01}01 \rangle$. This however, is not the first lexicographical embedding of $\langle 001 \rangle$ in $\langle 00101 \rangle$ (which is $\langle 001 {\color{blue} 01} \rangle$). Thus, if we wish to construct a lexicographically first embedding of $\ssub$ in the resulting string $s$ then in the last step, where we apply  \ref{first_type_ins} to several runs, we should never apply \ref{first_type_ins} to the first bits resulting from the application of  \ref{second_type_ins_0} in the first step.
	
	In view of the above discussion we say that a $1$-run is \emph{frozen} if it is the first bit of a substring that resulted from applying  \ref{second_type_ins_0}. In other words, a $1$-run is not frozen if it is either an original $1$-run of $\ssub$, a $1$-run that was added in the second step, or if it is the $2$nd or $3$rd bits  generated by applying \ref{second_type_ins_0}  (i.e. if we had $\langle 001 \rangle \rightarrow \langle 0101\rangle $ then the  non-frozen $1$-runs are the blue-colored bits $ \langle 0{\color{blue} 10}{1} \rangle$, and the last bit may also be non-frozen).
	

	Let $r_1$ and $r_2$ be the number of $1$-runs and $2$-runs in $\ssub$ and let $x$ be an integer such that $0\leq x \leq \delta m$.\\
	
	\begin{algorithm}[H] \label{alg:1}
		\SetAlgoLined
		\DontPrintSemicolon
		
		\SetKwInOut{Input}{input}
		\SetKwInOut{Output}{output}
		\SetNlSty{large}{[}{]}
		\Input{$\ssub \in S$ such that $\left|\ssub\right| = m - \delta m$, \newline
			and $0\leq x\leq \delta m$}
		\Output{A string $s\in S_{m,\beta_1}$ such that $\ssub$ is a subsequence of $s$}
		Select $x$ $2$-runs in $\ssub$ and apply \ref{second_type_ins_0} to them. \label{step:1} \;
		\tcc{total number of $1$-runs is $r_1+3x$ and of $2$-runs is $r_2-x$} \tcc{total number of non-frozen $1$-runs is $r_1+2x$} 
		Add  $\beta  m - r_1 - r_2 - 2x$ many $1$-runs to the right of the string  \label{step:2}\;
		\tcc{total number of runs is $\beta m$ and number of $1$-runs is $ \beta m - r_2+x$} 
		\tcc{total number of non-frozen $1$-runs is}
		\tcc{$\beta  m - r_1 - r_2 - 2x+ r_1+2x = \beta m -r_2$}
		Select $\delta m - (\beta m - r_1 - r_2 - x)$ non-frozen $1$-runs and apply \ref{first_type_ins} to each of them \;
		\tcc{length of resulting string is exactly $m$ \label{step:3}}
		\caption{Embed}\end{algorithm}
	\begin{claim}
		Algorithm \ref{alg:1} returns a string in $S_{m,\beta_1}$.
	\end{claim}
	
	\begin{proof}
		Step~\ref{step:1} turns each of the chosen $x$ $2$-runs into three $1$-runs, only two of which are non-frozen. Hence, the number of $2$-runs is $r_2-x$, the number of  $1$-runs is $r_1+3x$ and the number of non-frozen $1$-runs is $r_1+2x$.\\
		
		Step~\ref{step:2} completes the number of runs to $\beta m$ by introducing $\beta  m - r_1 - r_2 - 2x$ new $1$-runs. The total number of $2$-runs did not change, the total number of $1$-runs is now
		$$(r_1+3x)+(\beta  m - r_1 - r_2 - 2x) = \beta m -r_2+x$$
		and the number of non-frozen $1$-runs is
		$$(r_1+2x)+(\beta  m - r_1 - r_2 - 2x) = \beta m -r_2 \;.$$
		
		Step~\ref{step:3}  turns $\delta m - (\beta m - r_1 - r_2 - x)$ $1$-runs into $2$-runs. We now show that this gives a string in $S_{m,\beta_1}$. For this we need to show that it has only $1$-runs and $2$-runs and the correct number of runs of each type. The fact that we only get $1$- and $2$-runs follows from the definition of our operations. Now, the resulting number of $1$-runs is
		\begin{eqnarray*}
			(\beta m -r_2+x) - (\delta m - (\beta m - r_1 - r_2 - x)) &=& 2\beta m -\delta m - 2r_2 - r_1 \\ 
			&=& 2\beta m -\delta m - (m-\delta m) \\ &=& 2\beta m -m \\&=& \beta_1 m \;,
		\end{eqnarray*}
		where we have used the facts that $m-\delta m= |\ssub| = 2r_2+r_1$, that $\beta = \beta_1+\beta_2$ and that $m = \beta_1 m + 2\beta_2 m$. Similarly, the number of $2$-runs is
		\begin{eqnarray*}
			(r_2 - x) + (\delta m - (\beta m - r_1 - r_2 - x)) &=& r_1 + 2r_2 + \delta m - \beta  m \\ &=& m-\delta m + \delta m - (\beta_1 + \beta_2)m \\ &=& \beta_2 m \;.
		\end{eqnarray*}
		Note that by our construction, the string begins with a $1$ and it also ends with a $1$ as the total number of runs is odd (recall Remark~\ref{rem:odd}).
		Thus, the resulting string is in $S_{m,\beta_1}$ as claimed. 
	\end{proof}
	
	The next claim shows that any $s\in S_{m,\beta_1}$ that contains $\ssub$ as a subsequence can be obtained from the algorithm for an appropriate choice of $0\leq x \leq \delta m$.
	
	\begin{claim}
		For any $s\in S_{m,\beta_1}$ that contains $\ssub$ as a subsequence, there exists an $0\leq x \leq \delta m$ and appropriate choices  for the different steps of the algorithm so that the resulting string is $s$.
	\end{claim}
	
	\begin{proof}
		The claim can be proven by a simple induction on the length of $\ssub$ by e.g., considering the way that the first run (or second run in case that the length of the first run is one) in $\ssub$ is embedded in $s$, in the first lexicographical embedding of $\ssub$ in $s$. As the proof is simple we leave the details to the reader.
	\end{proof}

	To conclude, if we consider all possible values $x$ can take, and all the possibilities to perform the choices in the algorithm we get an upper bound on the number of strings $s\in S_{m,\beta_1}$ which contain $\ssub$ as a subsequence. We are now ready to prove Proposition~\ref{inner_insertion_prop}.

	\begin{proof}[Proof of Proposition~\ref{inner_insertion_prop}]	
		By the argument above it is enough to count the number of possibilities for $x$ and the number of possible choices made by the algorithm. For any choice of $x$, there are exactly ${r_2 \choose x}$ ways of selecting $x$ many $2$-runs in Step~\ref{step:1} of Algorithm~\ref{alg:1}. In Step~\ref{step:2} of the algorithm we have no freedom since we add the new $1$-runs at the end of the string. Finally, in Step~\ref{step:3} we have 
		${\beta m -r_2 \choose \delta m - (\beta m - r_1 - r_2 - x)}$ many ways to choose  $\delta m - (\beta m - r_1 - r_2 - x)$ many $1$-runs among the non-frozen $1$-runs. 
		
		Hence, the total number of strings that can be obtained from the algorithm is upper bounded by 
		%
		\begin{eqnarray*}
			\sum_{x=0}^{\delta m} \binom{r_2}{x} \binom{ \beta m - r_2}{\delta m - (\beta m - r_1 - r_2 - x)} 
			&\leq & \sum_{x=0}^{\delta m} \binom{\beta_1 m + \beta_2 m}{\delta m} \\
			&= & O (\delta m)  \cdot \binom{\beta_1 m + \beta_2 m}{\delta m} \;.
		\end{eqnarray*}
	\end{proof}
	
	Armed with Propositions~\ref{deletion_prop} and~\ref{inner_insertion_prop} we now show the existence of an appropriate inner code.
	
	\begin{prop} \label{inner_code_lem}
		Let $0\leq \beta_1 ,\delta\leq 1$  be parameters. Let $\beta = \frac{1+\beta_1}{2}$. For every $\varepsilon>0$ there is $M_\varepsilon$ so that for every  $m>M_\varepsilon$ there is a set $C \subseteq S_{m,\beta_1}$ of size $\left| C \right| = 2^{m\Rin}$ where
		\[
		\Rin = \beta  h\left( \frac{\beta_1}{\beta } \right) - (\delta + \beta ) h\left( \frac{\delta }{\delta + \beta } \right) - \beta  h\left( \frac{\delta }{\beta } \right) - \varepsilon \;,
		\]
		such that for every $c\neq c' \in C$ it holds that any string $\ssub \in S$ that is a subsequence of both  $c$ and $c'$ is of length $\left| \ssub \right| < m - \delta m$.	
	\end{prop}
	
	\begin{proof}
		We first note that the number of binary strings in $S_{m, \beta_1}$ is exactly $\binom{\beta m}{\beta_1m}$ as we have $\binom{\beta m}{\beta_1m}$ ways to arrange the $\beta_1 m $ $1$-runs and the $\beta_2 m$ $2$-runs. 
		
		The construction of $C$ is done greedily. We go over all strings $s\in S_{m,\beta_1}$ and add them to $C$ one by one as long as they do not share a too long common subsequence (from $S$) with any string that is already in $C$. Propositions~\ref{deletion_prop} and~\ref{inner_insertion_prop} imply that any $s\in S_{m,\beta_1}$ contains at most $\binom{\beta m + \delta m}{\delta m}$ many subsequences of length $m-\delta m$ from $S$, and each such string is a subsequence of at most  $O(\delta m)\binom{\beta  m } {\delta m}$  strings in $S_{m,\beta_1}$. Thus, whenever we add a string to $C$ we exclude at most 
		$$ \binom{\beta m + \delta m}{\delta m} \cdot O(\delta m)\binom{\beta  m } {\delta m}$$
		other strings from being in $C$. Therefore, our codebook contains at least 
		\[
		|C| \geq 
		\frac{\binom{\beta  m}{\beta_1 m}}{O(\delta m) \binom{\beta m + \delta m}{\delta m} \binom{\beta  m  } {\delta m}} 
		\geq 2^{m \left( \beta  h\left( \frac{\beta_1}{\beta } \right) - (\delta + \beta ) h\left( \frac{\delta }{\delta + \beta } \right) - \beta  h\left( \frac{\delta }{\beta } \right) \right)-O(\log m)}
		\]
		codewords, where the inequality follows by Lemma \ref{lemma:stirling}.
		Thus, for every $\varepsilon > 0$ there exists large enough $m>0$ such that the constructed set $C \subset S_{m, \beta_1}$ is of size  $2^{m \Rin}$ where $\Rin = \beta  h\left( \frac{\beta_1}{\beta } \right) - (\delta + \beta ) h\left( \frac{\delta }{\delta + \beta } \right) - \beta  h\left( \frac{\delta }{\beta } \right) - \varepsilon$.
	\end{proof}
	
	By construction, the code $C$ can handle an adversary that, given a codeword $c\in C$, returns a subsequence $\ssub \in S$ of $c$ where $\left| \ssub \right| \geq m - \delta m$. That is, we can uniquely identify the original codeword $c$ from $\ssub$.
	Our next goal is to show that our code can handle the usual edit distance adversary. In other words, it can handle an adversary that performs any $\delta m$ insertion and deletion (and hence it is not bound to return a string in $S$).
	The key observation is that if look at two different codewords $c,c'\in C$ and denote by $s$ a longest common subsequence between $c$ and $c'$ then it must be that there exists $s'\in S$ that is also a subsequence of $c$ and $c'$ and $\left|s\right| = \left|s'\right|$.
	
	
	\begin{prop}
		Let $C$ be the code constructed in Proposition~\ref{inner_code_lem}. 
		For any two codewords $c,c' \in C$ it holds that $\ed(c,c') > 2\delta m$.
	\end{prop}
	
	\begin{proof}
		Let $c \neq c' \in C$. Lemma~\ref{lem:lcs} gives
		$$\textup{ED}(c,c') = \left|c\right| + \left|c'\right| - 2\left| \textup{LCS}(c,c') \right| = 2m - 2\left| \textup{LCS}(c,c') \right| \;.$$
		Observe that if $s$ is a longest common subsequence of $c$ and $c'$ then there is a string $\ssub \in S$, such that $|\ssub|=|s|$ and $\ssub$ is also a common subsequence of $c$ and $c'$. Indeed, let $s$ be a longest common subsequence of $c$ and $c'$. Since both $c$ and $c'$ start and end with $1$, $s$ also starts and ends with $1$. Moreover, for every three consecutive, equal bits in $s$, we can flip the second bit and the resulting string will still be a common subsequence of maximal length (as neither $c$ nor $c'$ contain a run of length $3$ or more). Repeating this we will get a string only containing $1$-runs and $2$-runs, i.e. a string in $S$.
		
		As $C$ was constructed so that not two codewords in $C$ share a common subsequence (from $S$) of length larger or equal to $m-\delta m$, it follows that 
		$$	\textup{ED}(c, c') = 2m - 2\left| \textup{LCS}(c,c')\right| > 2m - 2(m-\delta m) = 2\delta m \;.$$
	\end{proof}
	
	
	\begin{remark}\label{rem:inner}
		Note that $C$ can be constructed in time at most $O\left(2^{2m}\cdot m^2\right)$ as in the worst case we compute the edit distance between any two possible strings.
	\end{remark}
	
	\section{Construction}
	\label{sec:construction}
	In this section we give a construction of a code for the BDC$_p$. Throughout this section we fix $p$.
	
	We repeat the high level description of the construction from Section \ref{sec:intro-const} (and as depicted in Figure \ref{fig:M1} on page \pageref{fig:M1}). We first do code concatenation. As outer code we use the one given in \cite[Theorem 1.1]{haeupler2017synchronization} (restated as Theorem \ref{Haupler_thm} here). As the inner code we use the code constructed in Proposition~\ref{inner_code_lem}. Then, in order to protect the concatenated codeword from a large number of deletions, we first place a \emph{buffer} of zeros between every two consecutive inner codewords. Since the decoder first looks for the buffers in order to identify where an inner code starts and where it ends, this step helps to reduce the amount of synchronization errors in the outer code.
	Secondly, we \emph{blow-up} the inner codewords by replacing every run of length $1$ with a run of length $N_1$ and every run of length $2$ with a run of length $N_2$, where the symbols of the runs are preserved. For example, $\langle 11 \rangle$ turns into $\langle  1^{N_2}\rangle$ and $ \langle 0 \rangle$ is replaced with $\langle 0 ^{N_1}\rangle$. If we choose $N_1$ and $N_2$ appropriately, then (with high probability) the decoder will identify the original run length.\\
	
	We now give a formal description of our construction. 
	
	\paragraph{The parameters:}
	At this point, we do not specify the parameters explicitly. We prefer to first present the scheme and analyze it before optimizing the parameters. However, the order by which we choose the parameters is important as there are some dependencies among them.
	First, we choose $M_1, M_2, \beta_1, M_B, \delout$ to be fixed constants. One should have in mind that $M_1  < M_2$ are the quantities by which we blow-up the different types of runs. 
	Then we choose $\delin$ to be larger than some quantity $\gamma=\gamma(M_1, T, M_2, \beta_1)$ that we later define (see Proposition~\ref{prop:decode}). At this point, we can compute the value of $\Rin$, the rate of the inner code, using Proposition \ref{inner_code_lem}. 
	Then, we choose a small enough $\epsout$ that determines the alphabet size of the outer code $\Cout$. 
	Denote with $\Rout$ the rate of  $\Cout$ and with $n$ its block length.
	Finally, we pick $m$, the block length of the inner code, to satisfy $ \Sigma   = \{0,1\}^{m \cdot \Rin}$.\footnote{When we choose parameters we make sure that $m \cdot \Rin$ is an integer.}
	
	While this may seem a bit confusing the main thing to remember is that $\epsout$ that was picked at the end, can be taken to be as small a constant as we wish, or, in other words, we can pick $m$ to be as large a constant as we wish. This is important as we will bound the probabilities of several bad events by expressions of the form $\exp(-\Omega(m))$ and it will be important for us to be able to pick $m$ large enough as to make all our estimates small.

	\paragraph{Encoding:}
	The process of encoding starts with the outer code. Given as input a message $x\in \Sigma^{\Rout n}$, we encode it with the code given in Theorem~\ref{Haupler_thm} to obtain an outer codeword $c^{(\textup{out})} = (\sigma_1, \ldots, \sigma_n) \in C_{\textup{out}} \subset \Sigma^n$. Then, every symbol in $c^{(\textup{out})}$, $\sigma_i \in \Sigma = \{0,1\}^{m \cdot \Rin}$, is encoded using the inner code to a codeword that we denote  $c^{(\textup{in})}_{\sigma_i}$. We thus get a codeword in the concatenated code 
	\[ 
	\left( c^{(\textup{in})}_{\sigma_1},\ldots, c^{(\textup{in})}_{\sigma_n} \right) \in C_{\textup{out}} \circ C_{\textup{in}} \;.
	\]
	Now that we have a codeword in the concatenated code we add additional layers of encoding that are crucial for the decoding algorithm to succeed.
	
	\begin{enumerate}
		\item Every two adjacent inner codewords are separated by a buffer of zeros of length $\ceil{M_B\cdot m/(1-p)}$.
		\item In every inner codeword, we replace every $1$-run with a run of length $\ceil{M_1/(1-p)}$ where the symbol of the run is preserved.
		\item In every inner codeword, we replace every $2$-run with a run of length $\ceil{M_2/(1-p)}$ where the symbol of the run is preserved.
	\end{enumerate}
	
	\sloppy After the buffering and blow-up process we have three different run lengths $\ceil{M_B\cdot m/(1-p)}, \ceil{M_1/(1-p)}$ and $\ceil{M_2/(1-p)}$. 
	Note that the buffer's length is much larger than $\ceil{M_1/(1-p)}$ and $\ceil{M_2/(1-p)}$ since it grows with $m$.

	\paragraph{Block length and rate:}
	
	Note that as $\Cin$ contains strings in $S_{m,\beta_1}$, each of the $n$ inner codewords becomes of length 
	$ \ceil{ M_1/(1-p) } \cdot \beta_1 m  + \ceil{ M_2/(1-p)} \cdot \beta_2 m $. As we have $n-1$ buffers between codewords the total block length is 
	$$\left(\ceil{ M_1/(1-p) } \cdot \beta_1 m  + \ceil{ M_2/(1-p)} \cdot \beta_2 m\right)\cdot n + \ceil{M_B\cdot m/(1-p)} \cdot (n-1) \;.$$
	Since the input to the encoding is a string  in $\Sigma^{ \Rout n}$ the rate $\mathcal{R}$ of the construction is given by 
	\begin{eqnarray} 
	\mathcal{R} &=& \frac{\log\left(\left|\Sigma\right|^{ \Rout n}\right)}{\beta_1 \ceil{ M_1/(1-p) } m n + \beta_2\ceil{ M_2/(1-p)} m n + \lceil M_B m/(1-p)\rceil (n-1)} \nonumber \\
	& \geq & \frac{\Rin \Rout}{\beta_1 \ceil{M_1/(1-p)} + \beta_2 \ceil{M_2/(1-p)} + M_B/(1-p) + 1/m} \nonumber \\
	& \geq & \frac{\Rin \Rout}{\beta_1 M_1/(1-p) + \beta_2 M_2/(1-p) + \beta + M_B/(1-p) + 1/m} \label{Rate_BDC_eq} \\
	& = & \frac{\Rin \Rout (1-p)}{\beta_1 M_1 + \beta_2 M_2 + \beta (1-p)+ M_B + (1-p)/m} \;. \nonumber
	\end{eqnarray}
	
	We can avoid the ceilings if we consider values of $p$ such that $\ceil{M_1/(1-p)},\ceil{M_2/(1-p)}$ and $\lceil M_B m/(1-p) \rceil$ are integers. 
	In this case, the rate is 
	\begin{equation} \label{Rate_BDC_eq_no_ceil}
	\mathcal{R} \geq \frac{\Rin \Rout (1-p)}{\beta_1 M_1 + \beta_2 M_2 + M_B} \;.
	\end{equation}
	
	
	\paragraph{Run time analysis:}
	
	By Theorem~\ref{Haupler_thm}, the outer code can be constructed in linear time. 
	Constructing the inner code requires time at most $O\left( 2^{2m}\cdot m^2 \right)$ (see Remark~\ref{rem:inner}).
	As $m = \log\left| \Sigma \right| / \Rin$, we get that constructing the inner code takes time $O(m^2 \cdot 2^{2 m}) = \left| \Sigma \right|^{O(1)}  = O_{\epsout}(1)$, which is constant.
	Thus, as all encoding steps are done in linear time, the encoding time complexity is $O(n)$.
	
	%

	\section{Correctness and Analysis}\label{sec:analyze}
	We first present the decoding algorithm and then prove its correctness. 
	After that, we show how to choose the parameters to obtain Theorem~\ref{BDC_theorem}.
	
	Let $y$ be the binary string received after transmitting Enc$(x)$. The decoding procedure is given in Algorithm~\ref{alg:decode} in page \pageref{alg:decode}. Observe that the algorithm depends on some integral parameter $T$. When analyzing the algorithm we will see what $T$ has to satisfy in order for the algorithm to decode successfully with high probability. For the time being it is enough to remember that $M_1 < T < M_2$.\\

	Before proving the correctness of the algorithm we give its run time analysis.

	\paragraph{Run time analysis of Algorithm~\ref{alg:decode}.}
	It is clear that Steps~\ref{decode:buffer} and \ref{decode:blow-up:end} take linear time. Step~\ref{decode:inner} runs the inner decoding algorithm $n$ times. As the inner decoding algorithm is a brute force that is run on strings of constant length it takes constant time. Thus, the first three steps of the decoding algorithm require linear time. In Step~\ref{decode:step:outer} we run the decoding algorithm of  \cite{haeupler2017synchronization} (recall Theorem~\ref{Haupler_thm}), that requires $O(n^2)$ time. 
	Therefore, the entire decoding procedure is dominated by the last step which runs in time $O(n^2) $.

	

	\begin{figure*}
		\begin{algorithm}[H]\label{alg:decode}
			\SetKwInOut{Input}{input}
			\SetKwInOut{Output}{output}
			\SetNlSty{textbf}{[}{]}
			\SetAlgoNlRelativeSize{-1}
			\LinesNumberedHidden
			\DontPrintSemicolon
			\Input{ Binary string $y$ which is the output of the BDC$_p$ on $\text{ENC}(x)$}
			\Output{A message $\tilde{x}\in \Sigma^k$}
			\SetAlgoLined
			
			\nlset{1} \tcc{Identifying buffers Step:}
			Every run of zeros of length longer than $M_B\cdot  m /2$ is identified as a buffer.  \;
			\label{decode:buffer}
			\tcc{Denote by $s_1,\ldots,s_t$ the strings between the identified buffers.}
			\nlset{2} \tcc{Threshold decoding step:}
			\For{every $s_i$}{
				\For {every run in $s_i$ \label{decode:blow-up}}{
					\eIf{the length of the run is longer than $T$}
					{Decode it to a run of length $2$}
					{Decode it to a run of length $1$}
			}}\label{decode:blow-up:end}
			\tcc{Let  $\tilde{c}_1, \ldots, \tilde{c}_t$ be the strings obtained in this step.}
			
			\nlset{3} \tcc{Inner code decoding step:} 
			Use brute-force decoding to decode each $\tilde{c}_i$ to get $\tilde{\sigma}_i$. Denote $\tilde{\sigma}^{\sf out}=(\tilde{\sigma}_1,\ldots,\tilde{\sigma}_n)\in \Sigma^n$ \label{decode:inner} \; 
			
			\nlset{4} \tcc{Outer code decoding step:}
			Run the decoding algorithm of the outer code on $\left(\tilde{\sigma}_1, \ldots, \tilde{\sigma}_t\right)$ to obtain $\tilde x$ \label{decode:step:outer} \;
			Output $\tilde x$ \;
			\caption{Decode with threshold $T$}
		\end{algorithm}
		\caption{Algorithm for decoding our code over BDC$_p$. The algorithm is assumed to know the parameters $k,n,m,M_B,T$ as well as $\Cin$ and $\Cout$.}
	\end{figure*}
	
	\subsection{Correctness of Decoding Algorithm}\label{sec:decode-analyze}
	
	In this section, we prove that Algorithm~\ref{alg:decode} succeeds with high probability.
	\begin{prop}\label{prop:decode}
	Given $M_1, T, M_2, M_B, \beta_1, \delin, \epsin, \delout$ (as described in Section \ref{sec:construction}) let
	$Z_1 \sim \textup{Bin}(\ceil{M_1/(1-p)}, 1-p)$ and $Z_2 \sim \textup{Bin}(\ceil{M_2/(1-p)}, 1-p)$. 
	Denote 
	\begin{align*}
	&P^{(1) \rightarrow (2)} := \Pr[Z_1 \geq T + 1], \\
	&P^{(1)\rightarrow (0)} := \Pr[Z_1 = 0], \\ 
	&P^{(2) \rightarrow (1)} := \Pr[Z_2 \leq T], \\
	&P^{(2)\rightarrow (0)} := \Pr[Z_2 = 0] \;,
	\end{align*}
	
		 and define
		\begin{equation} \label{eq:gamma-bdc}
		\gamma := \beta_1 \cdot P^{(1) \rightarrow (2)} + \beta_2 \cdot P^{(2) \rightarrow (1)} + \left(2\beta_1+ \beta_2 \right) P^{(1)\rightarrow (0)}+ 4 \beta_2 P^{(2)\rightarrow (0)} 
		\end{equation}
		(the reason for this definition of $\gamma$ is revealed later). Let $x\in\Sigma^{\Rout n}$ be a message and let $y$ be the string obtained after encoding $x$ using our code and transmitting it through the BDC$_p$. 
		If $\gamma < \delin$, then there exists $\epsilon_0=\epsilon_0(M_1, T, M_2, M_B,\beta_1, \delin, \delout)$ such that for every $\epsout < \epsilon_0$ it holds that Algorithm~\ref{alg:decode} returns $x$ with probability $1-\exp\left(-\Omega(n)\right)$.
	\end{prop}
	
	Observe that as $\epsilon_0=\epsilon_0(M_1, T, M_2, M_B,\beta_1, \delin, \delout)$, it does not depend on $m$ and $n$. 
	The rest of Section~\ref{sec:decode-analyze} is devoted to proving Proposition~\ref{prop:decode}. We first discuss the structure of the proof and prove relevant lemmas. The actual proof is given at the end of this section.\\
	
	
	Let $\sigma^{\sf out}=(\sigma_1,\ldots,\sigma_n)\in \Sigma^n$ be the result of encoding $x$ with the outer code. I.e. the first step before concatenating with our inner code. Let $c^{(\sf in)}_{\sigma_i}$ be the result of encoding $\sigma_i$ with the inner code.
	
	The decoding algorithm succeeds if the decoding procedure of the  outer code, which is executed in Step~\ref{decode:step:outer} of the algorithm, outputs the correct message. This happens if ED$(\sigma^{(\sf out)}, \tilde{\sigma}^{(\sf out)}) \leq \delout n$. To prove that this holds with high probability, we classify the errors that can be introduced at each step of the algorithm and bound the probability that we get too many of them.
	
	There are three error types that increase the edit distance between $\sigma^{(\sf out)}$ and $ \tilde{\sigma}^{(\sf out)}$:
	\begin{enumerate}
		\item \emph{Deleted buffer:} This happens when the channel deleted too many bits from a buffer so that less than  $M_B m/2$ bits survived the channel, and we did not identify this buffer in Step~\ref{decode:buffer} of the algorithm.
		
		\item \emph{Spurious buffer:}  In this case the algorithm mistakenly identifies a buffer inside an inner codeword. This might happen if there are many consecutive runs of the symbol $1$ that were deleted by the channel. As a result, a long run of the symbol $0$ is created and the algorithm will mistakenly identify it as a buffer in Step~\ref{decode:buffer}.
		
		\item \emph{Wrong inner decoding:} Here the decoding of the inner code returns a different inner codeword. This error happens if the edit distance between an inner codeword $c^{(\sf in)}_{\sigma_i}$ and the corresponding $\tilde{c}_j$ is larger than $ \delin m$.\footnote{Note that it may be the case that due to decoding errors, the $i$th inner codeword was interpreted as the $j$th codeword by the decoder (e.g. if a buffer was deleted or a spurious buffer was introduced).}
	\end{enumerate}
	
	In the following subsections, we analyze each error type separately and show that each happens with probability $\exp(-\Omega( m))$ per inner codeword. 
	Our analysis of the first two error types is similar to \cite{guruswami2017efficiently}, but our analysis of the third case is different. 
	
	\subsubsection{Deleted Buffer}
	\begin{prop} \label{clm:dele_buffer}
		Let $r_B$ be a buffer in $\text{Enc}(x)$. The probability that the decoding algorithm fails to identify it as a buffer in Step \ref{decode:buffer} is at most $\exp(-\Omega(m))$.
	\end{prop}
	\begin{proof}
		
		Recall that the length of a buffer is $\lceil M_B m/(1-p) \rceil$. Therefore the expected number of bits that survive the transmission trough the BDC$_p$ is at least $M_B m$. The decoder misses a buffer if the number of buffer bits that survived the transmission is smaller than $M_B  m/2$. 
		Let $Z$ denote the random variable that corresponds to the number of bits that survived the transmission of $r_B$  through the BDC$_p$. 
		Clearly, $Z\sim \text{Bin} (\lceil M_B m/(1-p) \rceil, (1-p))$.
		By using the Chernoff bound, we get that this error happens with probability  
		\[
		\Pr\left[Z < \frac{M_B m}{2}\right] = \Pr\left[Z < \left(1 - \frac{1}{2}\right)M_B m\right] < \exp\left( - \frac{1}{8} M_B m\right) \;.
		\]
	\end{proof} 
	
	\subsubsection{Spurious Buffer}
	Recall that this can happen if many consecutive runs of the symbol $1$ were deleted by the channel, so a long run of the symbol $0$ is created. If the length of this long run is longer than $M_B m/2$ then the decoder mistakenly identifies it as a buffer.
	
	\begin{prop} \label{clm:spuo_buffer}
		Let $c_{\sigma_i}^{(\sf in)}$ be an inner codeword. Denote by $\textup{Blow}(c_{\sigma_i}^{(\sf in)})$ the string obtained by blowing up the runs in $c^{(\sf in)}_{\sigma_i}$ according to the encoding procedure. The probability that the decoder in Step \ref{decode:buffer} identifies a buffer inside the string obtained by transmitting $\text{Blow}(c_{\sigma_i}^{(\sf in)})$ through the BDC$_p$ is at most $\exp(-\Omega(m))$.
	\end{prop}
	\begin{proof}
		We first compute the probability that a run is deleted. Recall that after encoding the message we transmit runs of length $\ceil{M_1/(1-p)}$ or $\ceil{M_2/(1-p)}$.
		The probability that all the bits from a run of length $\ceil{M_1/(1-p)}$ are deleted by the BDC$_p$ is
		\[
		p^{\ceil{M_1/(1-p)}} \leq p^{M_1/(1-p)} \leq e^{-M_1}.
		\]
		Equivalently, the probability that all the bits from a run of length $\ceil{M_2/(1-p)}$ are deleted by the BDC$_p$ is
		\[
		p^{\ceil{M_2/(1-p)}} \leq p^{M_2/(1-p)} \leq e^{-M_2}.
		\]
		Suppose that $\ell$ consecutive runs of the bit $1$ are deleted. We consider two cases.
		
		First, consider the case where $\ell >  m M_B/ 4M_2$. The probability that exactly $\ell$ runs of the symbol $1$ are deleted is at most (the highest probability is obtained when all the $\ell$ runs are $\ceil{M_1/(1-p)}$-runs)
		\[
		p^{\ceil{M_1/(1-p)}\ell} \leq \exp\left( -M_1 \ell \right) \leq \exp\left(- M_1M_Bm/4M_2\right) = \exp \left( - \Omega \left( m \right)\right).
		\] 
		The probability that there exist $\geq  m M_B/ 4M_2$ consecutive runs of the symbol $1$ that are deleted in a word of length $m$ is at most $O(m^2) \cdot \exp \left( -\Omega \left(m\right)\right) = \exp(-\Omega(m))$ (we just need to pick the start and end point of the consecutive runs).
		
		Now, if $\ell \leq m M_B/ 4M_2$ consecutive runs of $1$'s are deleted, then there are $\ell+1$ runs of zeros that are merged to a single run. 
		Suppose that all the merged runs were $2$-runs (so that the length of the run of the symbol $0$ that was created is maximized). Denote by $Z$ the random variable that corresponds to the number of bits that survived the transmission of these $\ell+1$ runs. It holds that $Z \sim \textup{Bin} \left( \left( \ell + 1 \right) \ceil{M_2 / \left(1 - p\right)}, 1 - p\right)$ and
		
		\begin{equation*} \label{Spurious_BDC_cons}
		\begin{split}
		\mathbb{E}\left[Z\right] &= \left( \ell + 1 \right)\ceil{M_2/(1-p)}(1-p)\\
		&\leq (\ell+1)(M_2 + 1)\\
		&\leq (m M_B/ 4M_2+1)(M_2 + 1)\\
		&\leq^{(*)} \frac{M_Bm+4M_2}{3} \\
		& < \frac{2}{5} M_B m
		\end{split}
		\end{equation*}
		where  inequality $(*)$ holds for $ M_2 \geq 3$ and the last inequality holds for large enough $m$.\footnote{Recall that by the way that we choose our parameters we pick $m$ at the end so that we can make it as large a constant as we wish.} Thus, we get by the Chernoff bound that the probability that $Z\geq M_Bm/2$ is
		\[
		\Pr\left[Z \geq \frac{M_B m}{2}\right] = \Pr\left[Z \geq \left(1 + \frac{1}{4}\right)\frac{2}{5}M_B m\right] \leq \exp\left( - \frac{1}{120} M_B m\right) \;.
		\]
		Hence, the probability that specific $\ell \leq mM_B/4M_2$ consecutive runs of the symbol $1$ were deleted and a spurious buffer was created is at most $ \exp( - M_B m/120 )$. Therefore, the probability that there exists a spurious buffer in an inner codeword of length $m$ is at most $m^2 \cdot \exp(- M_B m/120) \leq \exp(- M_B m/240)$, for large enough $m$.
	\end{proof}
	
	\subsubsection{Wrong Inner Decoding} \label{Inner_Dec_Analysis}
	This is the most difficult case to analyze. 
	The inner decoding procedure might output a wrong codeword when the edit distance between an inner codeword $c^{(\sf in)}_{\sigma_i}$ and the corresponding word that was obtained at Step~\ref{decode:blow-up} (and the two simple correction steps) of the algorithm, $\tilde{c}_j$, is larger than $\delta m$. The next proposition shows that the probability of this event is exponentially small in $m$.
	
	\begin{prop} \label{clm:wrong_inner_dec}
		Assume the setting of Proposition~\ref{prop:decode}.
		Let $c_{\sigma_i}^{(\sf in)}$ be an inner codeword. Assume that the buffers before and after $c^{(\sf in)}_{\sigma_i}$ were detected correctly and  that there were no spurious buffers in between. Suppose that $\tilde{c}_j$ is the corresponding string obtained at Step~\ref{decode:blow-up} of the decoding algorithm on $s_j$.\footnote{Recall that we use a different index $j$ to indicate that it may be the case that spurious buffers were found earlier, in some other inner codeword, or that some earlier buffers were mistakenly deleted.} Then,
		\[
		\Pr \left[ \textup{ED}\left(c^{(\sf in)}_{\sigma_i}, \tilde{c}_j\right)  > \delin m\right] \leq \exp(-\Omega (m)) \;.
		\]
	\end{prop}
	
	We prove this claim in the remainder of this subsection, but first we give some intuition and introduce some important notions.
	Recall that a run $r_j$ in an inner codeword is replaced with a run of length $ N_1 = \ceil{M_1/(1-p)}$ or $N_2 = \ceil{M_2/(1-p)}$. Let $Z_j$ be the random variable  corresponding to the number of bits from this blown-up run that survived the transmission through the BDC$_p$. If $\left| r_j \right| = 1$ then, $Z_j \sim \textup{Bin} \left( \ceil{ M_1/(1-p)}, 1-p\right)$. If $\left| r_j \right| = 2$ then, $Z_j \sim \textup{Bin} \left( \ceil{M_2/(1-p)}, 1-p\right)$. 
	Intuitively, in Step~\ref{decode:blow-up} the algorithm reads every $Z_j$ and decides according to the threshold $T$ if $Z_j$ corresponds to a run of length $1$ or $2$. However, it may be the case that, say, $Z_{j+1}=0$ and then the algorithm will mistakenly base its decision according to the value of $Z_j + Z_{j+2}$, etc. For example, consider an initial string $\langle 00100 \rangle$. After the blow-up, we transmit the string $\langle 0^{N_2} 1^{N_1}0^{N_2} \rangle$. Suppose that the middle run (the run consisting of the symbol $1$) was deleted by the channel. The decoder then faces a long run of $0$'s and treats it as a single run and in particular,  it will decode it as $\langle 0 \rangle$ or $\langle 00 \rangle$, or even as a spurious buffer. This motivates the following definitions.

	
	\begin{defi}
		When $Z_j = 0$ we say that $r_j$ was \emph{deleted by the channel}.
	\end{defi}
	
	\begin{remark}
		We shall make a distinction between runs that were deleted by the channel and those that our algorithm ``deleted'' so whenever we refer to a deleted bit we will stress which process caused the deletion.
	\end{remark}
	
	\begin{defi}\label{def:r'_j}
		For every $j\in [\beta  m]$, let $b_j$ be the bit appearing in $r_j$. We denote
		$$r_{j}' =  \left\{ \begin{array}{ll}
		\langle b_j b_j \rangle & \textrm{if $Z_j > T$}\\
		\langle  b_j \rangle & \textrm{if $0<Z_j \leq T$}\\
		\langle \rangle & \textrm{if $Z_j =0$}
		\end{array} \right. \;.
		$$  
	\end{defi}
	In other words, $r'_j$ is what Step~\ref{decode:blow-up} of our decoding algorithm would output when given $Z_j$ as input. In particular, $\left|r_{j}'\right|$ can be $0,1,2$, depending on $Z_j$. Note that if $\left|r'_j\right| = 0$ then it means that the channel deleted the run.

	
	For the next definition, we remind the reader that in our setting the total number of runs in an inner codeword (and hence also in a blown-up word) is $\beta_1 m + \beta_2 m = \beta m$.

	\begin{defi}
		A set $I \subset [\beta  m], \left|I\right| \geq 2$, is called a \emph{maximal merged set} if the following conditions hold:
		\begin{enumerate}
			\item For every $i\in I$ it holds that $Z_i>0$.
			\item All the bits from $I$ are merged into one run.
			\item There is no set $J$ such that $I \subsetneq J$ and the bits from $J$ are merged into one run. 
		\end{enumerate}
		
	\end{defi}
	For example, consider the following consecutive runs that were sent through the channel $\langle 0^{N_2}1^{N_1} 0^{N_1} 1^{N_2} 0^{N_1} 1^{N_1} 0^{N_2}\rangle$. Suppose that the third run and the fifth run were deleted by the channel and the rest of the runs were not deleted by the channel. The maximal merged set corresponding to this deletion pattern is $I = \lbrace 2,4,6\rbrace$. 
	
	\begin{claim} \label{deleted_in_merged_set}
		Let $I \subset [\beta  m]$ be a maximal merged set. Denote $j = \min I$ and $k = \max I$. Then, all the runs $r_{j+1}, r_{j+3}, \ldots ,r_{k-1}$ were \emph{deleted by the channel}.
	\end{claim}
	\begin{proof}
		Assume w.l.o.g. that $r_j$ and $r_k$ are runs of the symbol $0$. For every $i\in \lbrace j + 1, j+3, \ldots, k-1 \rbrace$, $r_i$ is a run of  symbol $1$ and must be deleted by the channel. Otherwise, $I$ will not be a merged set.
	\end{proof}

	\begin{defi}
		Let $I \subset [\beta  m]$ be a maximal merged set and set $j = \min I$. 
		We denote with $\tilde{r}_j$ be the result of Step~\ref{decode:blow-up} of our decoding algorithm on this merged run.  
	\end{defi}
	
	\begin{remark}\label{remark:different-r}
		It is important to remember that $r_j$ is the original run, $r'_j$ is what the algorithm would return when given $Z_j$ as input, and $\tilde{r}_j$ is what the algorithm actually returns when reading the bits of the merged run.
	\end{remark}
	
	We can now see that some bits that survived the channel were deleted by our algorithm as it failed to realize that they came from different runs. This is captured by the next definition.
	
	\begin{defi}\label{cla:deleted-bits-interval}
		Let $I \subset [\beta  m]$ be a maximal merged set. We say that the \emph{decoding algorithm deleted} $\left|r_{j}'\right| - \left|\tilde{r}_j\right| + \sum_{i\in I\setminus \lbrace j\rbrace} \left| r_i\right|$ bits in the set $I$.
	\end{defi}
	
	As $\left|r_{j}'\right| \leq \left|\tilde{r}_j\right|$ the following claim is obvious.
	
	\begin{claim} \label{merge_claim}
		Let $I$ be a maximal merged set and set $j = \min I$. The number of bits deleted by the decoding algorithm in the merged set $I$ is 
		at most $\sum_{i\in I\setminus \lbrace j\rbrace} \left| r_i\right|$.
	\end{claim}
	
	We next extend Claim~\ref{merge_claim} and bound the total number of bits that our algorithm deletes in an inner codeword. We assume that the buffers before and after the word were correctly identified by the decoding algorithm in Step \ref{decode:buffer}.
	
	\begin{claim} \label{deleted_runs_prop}
		Let $D \subset [\beta  m]$ be the indices of the runs that were deleted by the channel. If the last run was not deleted, i.e., $\beta  m \notin D$, then the number of bits that were deleted by the decoding algorithm is at most $\sum_{i\in D} \left|r_{i+1}\right|$.
		
		If the last run was deleted by the channel, i.e., $\beta  m \in D$, then the number of bits deleted by the algorithm is at most $\sum_{i\in D\setminus \lbrace \beta  m \rbrace}\left|r_{i+1}\right| 
		+ 2$.
	\end{claim}
	
	\begin{proof} 
		We first deal with the case where some runs were merged with the bits in the buffers (before or after the word). This happens if the first or the last run were deleted by the channel.
		If $1\in D$ then let $r_{i'}$ to be the first run of the symbol $1$ that was not deleted by the channel. Then, all runs of the symbol $0$ before $r_{i'}$ were merged to the left buffer. Therefore, $D_L := \lbrace 1, 3, \ldots , i' - 2 \rbrace \subseteq D$ and the decoding algorithm deleted exactly $\left|r_2\right| + \ldots +\left|r_{i' - 1}\right| = \sum_{\ell \in D_L}\left|r_{\ell + 1}\right|$ bits (since all these runs were considered as part of the buffer).
		
		Similarly, if $\beta  m\in D$ define $r_{i'}$ to be the last run of the symbol $1$ that was not deleted by the channel. In this case all the runs of $0$'s after $r_{i'}$ were merged to the right buffer. In this case, $D_R := \lbrace i' + 2, i' + 4, \ldots , \beta  m \rbrace \subseteq D$ and the decoding algorithm deleted exactly $\left|r_{i' + 1}\right| + \ldots +\left|r_{\beta  m - 1}\right| \leq 2+\sum_{\ell \in D_R\setminus \lbrace \beta  m \rbrace} \left|r_{\ell+1}\right|$ bits.
		
		We now account for inner deletions (i.e., those that did not cause runs to merge with buffers). These deletions may generate what we called maximal merged sets.
		Let $I_1, \ldots, I_t$ be all maximal merged sets, excluding those that were merged with buffers. 
		Denote $j_i = \min I_i$ and $k_i = \max I_i$ and let $ D_i := D\cap [j_i + 1, k_i - 1]$ for $ i\in [t]$.\footnote{Note that it may be the case that $D \setminus \left( \cup_i D_i\right)$ is not the empty set. In this case the indices in $D'$ did not cause a merge. E.g., if two consecutive runs are deleted by the channel and the runs before and after were not deleted.}
		
		According to Claim~\ref{deleted_in_merged_set} it holds that $\lbrace j_i + 1, j_i + 3, \ldots, k_i - 1\rbrace \subseteq D_i$. Thus, 
		$I_i \subseteq \lbrace j_i, j_i + 2, \ldots, k_i\rbrace$.
		Claim~\ref{merge_claim}, implies that the number of bits deleted by the algorithm in $I_i$ is at most $\sum_{\ell \in I_i\setminus \lbrace j_i\rbrace} \left| r_{\ell}\right|$. Thus, the total number of bits deleted by the algorithm, excluding those bits from $D_L\cup D_R$, is bounded from above by
		\[
		\sum_{i=1}^{t}\sum_{\ell \in I_i\setminus \lbrace j_i\rbrace}\left|r_{\ell}\right| \leq \sum_{i=1}^{t} \sum_{\ell \in D_i} \left|r_{\ell+1}\right| \leq \sum_{i\in D\setminus (D_L\cup D_R)} \left|r_{i+1}\right|.
		\]
		Taking into account the deleted bits from  $D_L\cup D_R$ the claim follows.
	\end{proof}
	
	We now use concentration bounds to argue about the expected number of bits that were deleted and the effect on the edit distance between the original inner codeword and the one returned by the algorithm in Step~\ref{decode:blow-up:end}.

	We first study the probability that $r_j\neq r'_j$ (recall Definition~\ref{def:r'_j}):
	\begin{enumerate}
		\item If $\left| r_j \right| = 1$ then there are two possible types of errors:
		\begin{enumerate}
			\item $ \left|r'_j\right| = 2$: We denote the probability for this to happen by $$P^{(1)\rightarrow (2)} := \Pr[Z_j \geq T + 1]\;.$$ We next give two estimates of this probability, one is an exact calculation and the other is an upper bound. Direct calculation gives
			\begin{equation}
			P^{(1)\rightarrow (2)} = \Pr[Z_j \geq T + 1] = \sum_{i=T+1}^{\ceil{\frac{M_1}{1-p}}} \binom{\ceil{\frac{M_1}{1-p}}}{i} (1-p)^i \cdot p^{\ceil{\frac{M_1}{1-p}} - i} \;.\label{small_to_big_run_probability}
			\end{equation}
			We next would like to use the Poisson distribution to give a simpler bound. For this we would like to use Theorems~\ref{bin_p_mitz} and \ref{cum_anderson}. 
%
			\begin{lemma} \label{lem:small_to_big_run_bound}
				Let $q \geq 1-p$. It holds that 
				\begin{equation}
				P^{(1)\rightarrow (2)} \leq 1 - e^{-M_1 - q} \sum_{i=0}^{T} \frac{(M_1 + q)^i}{i!} \label{small_to_big_run_bound} 
				\end{equation}
				Moreover, the function $f(q) := 1 - e^{-M_1 - q} \sum_{i=0}^{T} \frac{(M_1 + q)^i}{i!}$ is monotonically increasing in $q$.
			\end{lemma}
			\begin{proof}
				Define $Y(j,x) \sim \textup{Bin} (j, (M_1 + x)/j)$. Observe that $\mathbb{E}[Y(j,x)] = M_1 + x$.
				Denote $n' = \ceil{M_1/(1-p)}$. First note that 
				\[
				\Pr[Z_j \geq T + 1] \leq \Pr[Y(n', 1-p) \geq T + 1]
				\]
				since the expectation of $Y(n',(1-p))$ is $M_1 + (1-p)$ whereas the expectation of $Z_j$ is $\leq M_1 + (1-p)$ and they are both binomial distributions on $n'$ trials.	By the same reasoning we have that for every $j \geq n'$
				\[
				\Pr[Y(j,(1-p)) \geq T + 1] \leq \Pr[Y(j, q) \geq T + 1] \;.
				\]
				Let $P(x) \sim \textup{Poisson}(x)$. Theorem \ref{bin_p_mitz} implies that $\lim_{j \rightarrow \infty}Y(j, x) = P(M_1 + x)$.
				Therefore, 
				\begin{align*}
				P^{(1)\rightarrow (2)} = \Pr \left[ Z_j \geq T + 1\right] &\leq 
				\Pr[Y(n', q) \geq T + 1] \\
				&= 1 - \Pr \left[Y(n', q) \leq T \right] \\
				&\leq 1 - \lim_{j \rightarrow \infty}\Pr \left[Y(j, q)\leq T \right]\\
				&= 1 - \Pr \left[P(M_1 + q) \leq T \right] \\
				&= 1 - e^{-M_1 - q} \sum_{i=0}^{T} \frac{(M_1 + q)^i}{i!}\;,  
				\end{align*}
				where the second inequality follows from Theorem~\ref{cum_anderson} due to monotonicity for $T \geq M_1 + q$.
				
				Note that the monotonicity of $f(q) = 1 - e^{-M_1 - q} \sum_{i=0}^{T} \frac{(M_1 + q)^i}{i!}$ follows from Lemma \ref{lem:poisson_mono}.
			\end{proof}
			
%
			\item $ \left|r'_j\right| = 0$: Here the blown-up run was completely deleted by the channel. The probability for this to happen is $P^{(1)\rightarrow (0)} := \Pr [Z_j = 0]$. It holds that,
			\begin{equation}
			P^{(1)\rightarrow (0)} = \Pr[Z_j = 0] = p^{\ceil{\frac{M_1}{1-p}}}. \label{full_deletion_prob_1} 
			\end{equation}
			It also holds that for any $p \in (0,1)$,
			\begin{equation}
			P^{(1)\rightarrow (0)} = \Pr[Z_j = 0]  \leq e^{-M_1}.\label{full_deletion_bound_1}
			\end{equation}
		\end{enumerate}
		
		\item Similarly, when $\left|r_j\right| = 2$  there are two cases to consider:
		\begin{enumerate}
			\item $ \left|r'_j\right| = 1$: The probability for this to happen is $P^{(2)\rightarrow (1)} := \Pr[Z_j \leq T]$. As before, the exact calculation is
			\begin{equation} 
			P^{(2)\rightarrow (1)} = \Pr[Z_j \leq T] = \sum_{i=0}^{T} \binom{\ceil{\frac{M_2}{1-p}}}{i} (1-p)^i \cdot p^{\ceil{\frac{M_2}{1-p}} - i} \label{big_to_small_run_probability} \;. 
			\end{equation}
			
			Similarly to the calculations for	$P^{(1)\rightarrow (2)}$, we would like to upper bound $P^{(2)\rightarrow (1)}$ using a simpler expression coming from the Poisson distribution. 
			\begin{lemma} \label{lem:bounds-big-run-small-run}
				For every $p$, it holds that
				\begin{equation}
				P^{(2)\rightarrow (1)} \leq e^{-M_2} \sum_{i=0}^{T} \frac{M_2^i}{i!} \label{big_to_small_run_bound} \;.
				\end{equation}
				Moreover, for every $q \geq p$ such that $M_2/(1-q)$ is an integer, it holds that 
				\begin{equation}
				P^{(2)\rightarrow (1)} \leq \sum_{i=0}^{T} \binom{\frac{M_2}{1-q}}{i} (1-q)^i \cdot q^{\frac{M_2}{1-q} - i} \label{big_to_small_run_bound_q}
				\end{equation}
			\end{lemma}
			\begin{proof}
				For a natural number $1\leq i$, let $Y(i) \sim \textup{Bin} \left(i, M_2/i \right)$. Let $P \sim \textup{Poisson}(M_2)$. Observe that  $\Pr[Z_j \leq T] \leq \Pr[Y(\ceil{M_2/(1-p)}) \leq T]$ as the latter can only have smaller expectation. Since $\lim_{i\rightarrow \infty}Y(i) \sim P$ and due to the monotonicity implied by Theorem~\ref{cum_anderson} we get that when $T\leq M_2 -1$,  for every $p$ it holds that:
				\begin{align}
				P^{(2)\rightarrow (1)} = \Pr[Z_j \leq T] & \leq \Pr[Y(\ceil{M_2/(1-p)}) \leq T] \nonumber \\
				& \leq \Pr[Y(\lceil{M_2/(1-q)}\rceil) \leq T] \nonumber \\
				& = \Pr[Y(M_2/(1-q)) \leq T] \nonumber \\
				&\leq \lim_{i\rightarrow \infty}\Pr[Y(i) \leq T] \nonumber \\
				&= \Pr[P \leq T] \nonumber \\ 
				& = e^{-M_2} \sum_{i=0}^{T} \frac{M_2^i}{i!} \nonumber \;, 
				\end{align}
				where the second and the third inequalities hold due to Theorem \ref{cum_anderson} for $T \leq M_2 - 1$. Note that the second inequality proves the second statement in the lemma. 
			\end{proof}
%
			\item $ \left|r'_j\right| = 0$: The probability for this to happen is $P^{(2)\rightarrow (0)} := \Pr [Z_j = 0]$. It holds that,
			\begin{equation}
			P^{(2)\rightarrow (0)} = \Pr[Z_j = 0] = p^{\ceil{\frac{M_2}{1-p}}} \label{full_deletion_prob_2} 
			\end{equation}
			and for every $p\in(0,1)$ we have
			\begin{equation}
			P^{(2)\rightarrow (0)} = \Pr[Z_j = 0] \leq e^{-M_2} \;.\label{full_deletion_bound_2}
			\end{equation}
		\end{enumerate}
	\end{enumerate}


	Recall that $c_{\sigma_i}^{(\sf in)}$ is an inner codeword that consists of exactly $\beta_1 m$ $1$-runs and $\beta_2 m$ $2$-runs. Also recall that we blow-up an inner codeword, $c_{\sigma_i}^{(\sf in)}$, and send it through the BDC$_p$. Suppose that Step~\ref{decode:buffer} of the algorithm identified the $i-1$'th and the $i$'th buffer and that there were no spurious buffers in between. Let $s_j$ be the binary string corresponding to this decoding window obtained in Step \ref{decode:buffer}, and let $\tilde{c}_j$ be the result of Step~\ref{decode:blow-up} of the algorithm on $s_j$.
	
	For every $j\in [\beta m -1]$, Let $X_j$ be the random variable defined by 
	$$X_{j} =  \left\{ \begin{array}{ll}
	0 & \textrm{if $\left|r_j\right| = \left|r_j'\right|$}\\
	1  & \textrm{if $Z_j > 0$ and $\left|r_j\right|\neq \left|r_j'\right|$}\\
	\left| r_j \right| + \left|r_{j+1}\right| & \textrm{if $Z_j =0$}
	\end{array} \right. \;.
	$$
	Similarly define $X_{\beta m}$ to be
	$$X_{\beta m} =  \left\{ \begin{array}{ll}
	0 & \textrm{if $\left|r_{\beta m}\right| = \left|r_{\beta m}'\right|$}\\
	1  & \textrm{if $Z_{\beta m} > 0$ and $\left|r_{\beta m}\right|\neq \left|r_{\beta m}'\right|$}\\
	\left| r_{\beta m} \right| + 2 & \textrm{if $Z_{\beta m} =0$}
	\end{array} \right. \;.
	$$
	\begin{claim} \label{claim:ed_upper_bound}
		Let $c^{(\textup{in})}_{\sigma_i}$ be an inner codeword. Assume that the buffers before and after $c^{(\textup{in})}_{\sigma_i}$ were detected correctly and assume that there were no spurious buffers in between. Suppose that $\tilde{c}_j$ is the corresponding string obtained at Step~\ref{decode:blow-up} of the decoding algorithm on $s_j$.
		Then,
		\[
		\textup{ED}\left(c^{(\textup{in})}_{\sigma_i}, \tilde{c}_j\right) \leq \sum_{j=1}^{\beta m} X_j\;.
		\]
	\end{claim}
	\begin{proof}
		If $r_j$ is a $1$-run and $r_j'$ is a $2$-run then there was an insertion. Equivalently, if $r_j$ is a $2$-run and $r_j'$ is a $1$-run then there was a deletion. 
		If a run was completely deleted by the channel then according to Claim~\ref{deleted_runs_prop}, at the worst case scenario, the following run is also deleted by the algorithm. The definition of the $X_j$'s accounts for all that.
	\end{proof}
	Note that we may do over counting in some scenarios, e.g., if $r_{j+1} \neq r_{j+1}'$ and $r_j$ was deleted by the channel then $X_j + X_{j+1} = \left|r_j\right| + \left|r_{j+1}\right| + 1$ but the edit distance is at most $\left|r_j\right| + \left|r_{j+1}\right|$. This over counting makes the upper bound less tight.
	
	Set $X=\sum_{j=1}^{\beta m}X_j$. 
	We next upper bound and lower bound $\mathbb{E}[X]$.
	\begin{claim} \label{clm:x_expected_upper}
		It holds that
		\[
		\mathbb{E} \left[ X \right] \geq \xi m \;,
		\]
		where 
		\[
		\xi = \beta_1 \left(P^{(1)\rightarrow (2)} + 2\cdot P^{(1)\rightarrow (0)}\right) + \beta_2 \left(P^{(2)\rightarrow (1)} + 3\cdot P^{(2)\rightarrow (0)}\right)\;.
		\]
	\end{claim}
	\begin{proof}
		For every $X_j$ such that $r_j$ is a $1$-run we have 
		\[
		\mathbb{E}[X_j] \geq 1\cdot P^{(1)\rightarrow (2)} + 2 \cdot P^{(1)\rightarrow (0)} \;,
		\]
		where we used the fact that  $\left|r_j\right| + \left|r_{j+1}\right|\geq 2$.
		Similarly, for every $X_j$ such that $r_j$ is a $2$-run we have 
		\[
		\mathbb{E}[X_j] \geq 1\cdot P^{(2)\rightarrow (1)} + 3 \cdot P^{(2)\rightarrow (0)}\;.
		\]
		As there are exactly $\beta_1$ $1$-runs and $\beta_2$ $2$-runs, the claim follows.
	\end{proof}
	
	\begin{claim} \label{clm:x_expected_lower}
		It holds that
		\[
		\mathbb{E} \left[ X \right] \leq \gamma m + P^{(1)\rightarrow (0)} \;,
		\]
		where,
		\begin{equation} 
		\gamma = \beta_1 \cdot P^{(1)\rightarrow (2)} + \beta_2 \cdot P^{(2)\rightarrow (1)}
		+ \left(2\beta_1+ \beta_2 \right) \cdot P^{(1)\rightarrow (0)}
		+ 4 \beta_2 \cdot P^{(2)\rightarrow (0)} \;,
		\end{equation}
		is the same $\gamma$ as in Proposition~\ref{prop:decode}.
	\end{claim}
	
	For the proof we shall denote with $X_j^{i,k}$ the random variable $X_j$ when $r_j$ is an $i$-run and $r_{j+1}$ is a $k$-run.
	
	\begin{proof}
		Suppose that $r_{\beta m}$ is a $1$-run. As will be explained later, this is the worst case, i.e., the upper that we prove on $\mathbb{E}[X]$ is largest in this case.
		Denote by $Y^{i,k}\subseteq [\beta m-1]$ the set of indices $j\in [\beta m - 1]$ such that $r_j$ is an $i$-run and $r_{j+1}$ is a $k$-run.
		From linearity of expectation it follows that
		\[
		\mathbb{E}[X] = \sum_{j\in Y^{1,2}} \mathbb{E} \left[X_j^{1,2} \right] + \sum_{j\in Y^{1,1}} \mathbb{E} \left[X_j^{1,1} \right] + \sum_{j\in Y^{2,1}} \mathbb{E} \left[X_j^{2,1} \right] + \sum_{j\in Y^{2,2}} \mathbb{E} \left[X_j^{2,2} \right] + \mathbb{E}\left[X_{\beta m}\right] \; .
		\]
		Let $\lambda_1$ be  such that $\left|Y^{1,2}\right|= \lambda_1 m$. Thus, $\left|Y^{1,1}\right|=(\beta_1 - \lambda_1 )m - 1$ (where the $1$ is subtracted because of the last run, which we assumed is a $1$-run). By definition we have that
		\[
		\sum_{j\in Y^{1,2}} \mathbb{E} \left[X_j^{1,2} \right] = \left( 1\cdot P^{(1)\rightarrow (2)} + 3\cdot P^{(1)\rightarrow (0)}\right) \cdot \lambda_1 m 
		\]
		and 
		\[
		\sum_{j\in Y^{1,1}} \mathbb{E} \left[X_j^{1,1} \right] = \left( 1\cdot P^{(1)\rightarrow (2)} + 2\cdot P^{(1)\rightarrow (0)} \right) \cdot ((\beta_1 - \lambda_1) m - 1) \;.
		\]
		Observe that	
		\begin{eqnarray*}	
			& & \sum_{j\in Y^{1,2}} \mathbb{E} \left[X_j^{1,2} \right] + \sum_{j\in Y^{1,1}} \mathbb{E} \left[X_j^{1,1} \right] + \mathbb{E}[X_{\beta m}] \\ 
			&=& \left( 1\cdot P^{(1)\rightarrow (2)} + 3\cdot P^{(1)\rightarrow (0)}\right) \cdot \lambda_1 m + \left( 1\cdot P^{(1)\rightarrow (2)} + 2\cdot P^{(1)\rightarrow (0)} \right) \cdot ((\beta_1 - \lambda_1) m - 1)\\& & + \left( 1\cdot P^{(1)\rightarrow (2)} + 3\cdot P^{(1)\rightarrow (0)}\right) \\
			&=&  P^{(1)\rightarrow (2)} \cdot \beta_1 m  +  P^{(1)\rightarrow (0)} \cdot (2\beta_1 m + \lambda_1 m +1) \;.
		\end{eqnarray*}
		As there are exactly $\beta_2 m$ $2$-runs, it holds that $0\leq \lambda_1 \leq \beta_2$. Hence,  this sum is maximized for $\lambda_1 = \beta_2$. We thus have that
		\begin{equation}
		\begin{split}
		&	\sum_{j\in Y^{1,2}} \mathbb{E} \left[X_j^{1,2} \right] + \sum_{j\in Y^{1,1}} \mathbb{E} \left[X_j^{1,1} \right] + \mathbb{E}[X_{\beta m}] \\ \leq & \beta_1 m P^{(1)\rightarrow (2)} + (2\beta_1 + \beta_2)m P^{(1)\rightarrow (0)} + P^{(1)\rightarrow (0)}\;.
		\label{eq:sum1}
		\end{split}
		\end{equation}
		Similarly, let $\lambda_2$ be such that $\left|Y^{2,1}\right| = \lambda_2 m$. Thus, $\left|Y^{2,2}\right| = (\beta_2 -\lambda_2) m$. It holds that
		\[
		\sum_{j\in Y^{2,1}} \mathbb{E} \left[X_j^{2,1} \right] = \left( 1\cdot P^{(2)\rightarrow (1)} + 3\cdot P^{(2)\rightarrow (0)} \right) \cdot \lambda_2 m
		\]
		and 
		\[
		\sum_{j\in Y^{2,2}} \mathbb{E} \left[X_j^{2,2} \right] = \left( 1\cdot P^{(2)\rightarrow (1)} + 4\cdot P^{(2)\rightarrow (0)} \right) \cdot (\beta_2 - \lambda_2) m\;.
		\]
		Since the sum $\sum_{j\in Y^{2,1}} \mathbb{E} \left[X_j^{2,1} \right] + \sum_{j\in Y^{2,2}} \mathbb{E} \left[X_j^{2,2} \right]$ is maximized for $\lambda_2 = 0$ we get, 
		\begin{equation}
		\sum_{j\in Y^{2,1}} \mathbb{E} \left[X_j^{2,1} \right] + \sum_{j\in Y^{2,2}} \mathbb{E} \left[X_j^{2,2} \right] \leq \beta_2 m P^{(2)\rightarrow (1)} + 4 \beta_2 m P^{(2) \rightarrow (0)} \;.
		\label{eq:sum2}	
		\end{equation}
		Combining \eqref{eq:sum1} and \eqref{eq:sum2} we obtain
		\begin{eqnarray}
		\mathbb{E}\left[ X \right] &=& \sum_{j\in Y^{1,2}} \mathbb{E} \left[X_j^{1,2} \right] + \sum_{j\in Y^{1,1}} \mathbb{E} \left[X_j^{1,1} \right] + \mathbb{E}[X_{\beta m}] + \sum_{j\in Y^{2,1}} \mathbb{E} \left[X_j^{2,1} \right] + \sum_{j\in Y^{2,2}} \mathbb{E} \left[X_j^{2,2} \right]  \nonumber \\
		&\leq & \beta_1 \cdot P^{(1)\rightarrow (2)} + \beta_2 \cdot P^{(2)\rightarrow (1)}
		+ \left(2\beta_1+ \beta_2 \right) \cdot P^{(1)\rightarrow (0)}
		+ 4 \beta_2 \cdot P^{(2)\rightarrow (0)}+ P^{(1) \rightarrow (0)}  \nonumber\\
		&= &  \gamma m + P^{(1) \rightarrow (0)} \; ,	\label{eq:sum-all} 
		\end{eqnarray}
		as claimed. 
		
		Note that if $r_{\beta m}$ was a $2$-run, then $\left|Y^{1,2}\right| + \left| Y^{1,1}\right| = \beta_1 m$ (no need to subtract $1$ since the last run is now a $2$-run) and we have,
		\[
		\sum_{j\in Y^{1,2}} \mathbb{E} \left[X_j^{1,2} \right] + \sum_{j\in Y^{1,1}} \mathbb{E} \left[X_j^{1,1} \right] \leq \beta_1 m P^{(1)\rightarrow (2)} + (2\beta_1 + \beta_2)m P^{(1)\rightarrow (0)} \;.
		\]
		In this case, we have that $\left|Y^{2,1}\right| + \left|Y^{2,2}\right| = \beta_2 m -1$. Thus, if we let $\lambda_2 $ be such that $\left|Y^{2,1}\right| = \lambda_2 m$ and $\left|Y^{2,2}\right| = (\beta_2 - \lambda_2) m - 1$ then
		\begin{eqnarray*}	
			& & \sum_{j\in Y^{2,1}} \mathbb{E} \left[X_j^{2,1} \right] + \sum_{j\in Y^{2,2}} \mathbb{E} \left[X_j^{2,2} \right] + \mathbb{E}[X_{\beta m}] \\ 
			&=& \left( 1\cdot P^{(2)\rightarrow (1)} + 3\cdot P^{(2)\rightarrow (0)}\right) \cdot \lambda_2 m + \left( 1\cdot P^{(2)\rightarrow (1)} + 4\cdot P^{(2)\rightarrow (0)} \right) \cdot ((\beta_2 - \lambda_2) m - 1)\\& & + \left( 1\cdot P^{(2)\rightarrow (1)} + 4\cdot P^{(2)\rightarrow (0)}\right) \\
			&=&  P^{(2)\rightarrow (1)} \cdot \beta_2 m  +  P^{(2)\rightarrow (0)} \cdot (4\beta_2 m -\lambda_2 m) \;,
		\end{eqnarray*}
		and this sum is maximized for $\lambda_2 = 0$. We thus have that
		\[
		\sum_{j\in Y^{2,1}} \mathbb{E} \left[X_j^{2,1} \right] + \sum_{j\in Y^{2,2}} \mathbb{E} \left[X_j^{2,2} \right] + \mathbb{E}[X_{\beta m}]  \leq P^{(2)\rightarrow (1)} \beta_2 m  + 4 P^{(2)\rightarrow (0)} \beta_2 m
		\]
		Then, if $r_{\beta m}$ is a $2$-run we have
		\begin{equation*}
		\begin{split}
		\mathbb{E}[X] & \leq \beta_1 m P^{(1)\rightarrow (2)} + (2\beta_1 + \beta_2)m P^{(1)\rightarrow (0)} + P^{(2)\rightarrow (1)} \beta_2 m  + 4 P^{(2)\rightarrow (0)} \beta_2 m \\&= \gamma m <  \gamma m + P^{(1) \rightarrow (0)}  \;.
		\end{split}
		\end{equation*}
	\end{proof}
	
	Thus, for any constant $\gamma' > \gamma$ there exist a constant $M_{\gamma '}$ such that for all $m > M_{\gamma '}$ it holds that
	$$ \mathbb{E}[X] \leq \gamma m + P^{(1)\rightarrow (0)} < \gamma' m \;.$$
	In the following claim we use concentration bound to show that the probability that $X$ is greater than $\gamma' m$, for $\gamma'> \gamma$, is exponentially small in $m$ and then we conclude that decoding of an inner codeword succeeds with high probability.
	
	\begin{claim}\label{cla:upper-bound-gamma'}
		For any $\gamma' > \gamma$ and for every constant $\nu > 0$ it holds that for a large enough $m$,
		\[
		\Pr[X > (1 + \nu) \gamma' m] <  \exp\left( - \frac{\nu^2 \xi^2 m}{8 \beta} \right) = \exp(-\Omega(m)) \;,
		\]
		where $\xi$ is as in Claim~\ref{clm:x_expected_upper}.
	\end{claim}
	\begin{proof}
		First note that 
		\begin{equation*}
		\Pr[X > (1 + \nu) \gamma' m] < \Pr[X > (1 + \nu) \mathbb{E}[X]] \;,
		\end{equation*}
		where by Claim~\ref{clm:x_expected_lower} the inequality holds for large enough $m$. 
		The delicate point is to notice that the $X_j$'s are independent. This is because each $X_j$ is determined solely according to the value of $Z_j$ (indeed, its value only depends on whether $Z_j=0$, $Z_j\leq T$ or $Z_j > T$), and the random variables $Z_j$'s are independent by the definition of the binary deletion channel. 
		For every $X_j$ it holds that $ 0 \leq X_j \leq 4$ and if we set  $t = \nu \mathbb{E}[X]$ and apply Theorem~\ref{thm:hoeffding} then we get that
		\begin{equation*}
		\begin{split}
		\Pr \left[ X > (1+\nu) \mathbb{E}[X]\right] &< \exp \left( - \frac{2 \nu^2 (\mathbb{E}[X])^2}{\beta m \cdot 4^2} \right) \\
		& \leq \exp\left(-\frac{2 \nu^2 (\xi m)^2}{16 \beta m}\right) \\
		& = \exp\left( -\frac{\nu^2 \xi^2 m}{8 \beta} \right) \;,
		\end{split}
		\end{equation*}
		where the second inequality follow from Claim~\ref{clm:x_expected_upper}.
		
	\end{proof}
	
	We are now ready to prove the main claim of this subsection, Proposition~\ref{clm:wrong_inner_dec}.
	\begin{proof} [Proof of Proposition \ref{clm:wrong_inner_dec}]
		By Claim~\ref{claim:ed_upper_bound} $X$ is an upper bound on $\ed(c^{(\textup{in})}_{\sigma_i}, \tilde{c}_j)$. Thus, $$\Pr \left[ \ed\left(c^{(\textup{in})}_{\sigma_i}, \tilde{c}_j\right) > \delin m\right]\leq \Pr[X > \delin m] \;.$$
		By the assumption in Proposition~\ref{prop:decode} we have that $\delin > \gamma$.
		We thus get that
		$$ \Pr\left[X > \delin m\right]   = \Pr\left[X > \left(1+\frac{\delin - \gamma}{\delin + \gamma} \right) \frac{\delin + \gamma}{2} m \right]  \leq  \exp\left( -\left( \frac{\delin - \gamma}{\delin + \gamma}\right)^2 \frac{ \xi^2 }{8 \beta} m\right)  \;,$$
		where the last inequality follows from Claim~\ref{cla:upper-bound-gamma'} by plugging $\nu = \frac{\delin - \gamma}{\delin + \gamma}$ and $\gamma' = \frac{\delin + \gamma}{2} $. This completes the proof of Proposition~\ref{clm:wrong_inner_dec}.
	\end{proof}
	
	\begin{remark}\label{rem:m-ind}
		Observe that all the parameters involved in the upper bound in Proposition~\ref{clm:wrong_inner_dec}, namely, $\gamma,\xi,\beta$ are independent of $m$. That is, they only depend on $\delin, M_1,M_2,T$ and $\beta_1$. 
	\end{remark}

	We are now ready to prove Proposition \ref{prop:decode}.
	
	\begin{proof}[Proof of Proposition~\ref{prop:decode}]
		We would like to show that with high probability, the edit distance between the original outer codeword $\sigma^{(\textup{out})}$ and the string $\tilde{\sigma}^{(\textup{out})}$, obtained after Step~\ref{decode:inner} of the decoding algorithm, is smaller than $ \delta_{\textup{out}} n$. To prove this we shall analyze the contribution of each of the error types (deleted buffer, spurious buffer and wrong inner decoding) on the edit distance. 
		
		A deleted buffer causes two inner codewords to merge and thus be decoded incorrectly by the inner code's decoding algorithm. When considering the effect of this on the edit distance between $\sigma^{(\textup{out})}$ and $\tilde{\sigma}^{(\textup{out})}$, this introduces two deletions and one insertion. Similarly, a spurious buffer introduces one deletion and two insertions, since an inner codeword gets split into two parts. A wrong inner decoding causes just one deletion and one insertion.
		Therefore, every error type increases the edit distance between the original outer codeword $\sigma^{(\textup{out})}$ and $\tilde{\sigma}^{(\textup{out})}$ 
		by at most three. 
		
		As mentioned, the outer decoding algorithm fails if $\ed \left( c^{(\sf out)},\tilde{c}^{(\sf out)} \right)> \delta_{\textup{out}} n$. 
		Thus, for this to happen, at least one of the following bad events must occur:
		\begin{enumerate}
			\item There were at least $\delout n /9$ deleted buffers.
			\item There were at least $\delout n /9$ spurious buffers.
			\item There were at least $\delout n /9$ inner codewords that were decoded incorrectly even though they did not have spurious buffers and their buffers were identified.
		\end{enumerate}
		
		We saw in Propositions \ref{clm:dele_buffer},\ref{clm:spuo_buffer} and \ref{clm:wrong_inner_dec} that for every inner codeword, each error type happens with probability $\exp(- \Omega(m))$. 
		Since $\delout$ is a fixed constant, there exists a large enough $m$ so that $\exp(- \Omega(m)) \leq \delout/ 10$ for each error type. 
		An important observation is that, similarly to Remark~\ref{rem:m-ind}, the constants in the $\exp(- \Omega(m))$ in the different propositions depend only on $\delin, \beta_1, M_1, T, M_2, M_B$ which are fixed constants and are not related to the outer code. Thus, we can choose a small enough $\epsout$, which determines a large enough $m$, so that the probability for each error type is $\leq \delout/10$.
		
		By Chernoff bound, for a large enough $n$, each of the three bad events happens with probability $\exp(-\Omega (n))$. Thus, Algorithm~\ref{alg:decode} succeeds with probability $1-\exp(- \Omega(n))$.
	\end{proof}

	\subsection{Proof of Theorem~\ref{BDC_theorem}} \label{Parameters_BDC}
	
	We now prove our main theorem.

	\begin{proof}[Proof of Theorem~\ref{BDC_theorem}]
		
		Our goal is to maximize the rate given in Equation~\eqref{Rate_BDC_eq} while assuring that the parameters that we pick guarantee successful decoding with high probability. Recall that the order by which we choose the parameters in our construction is the following. First, we choose $M_1, T, M_2, \beta_1, M_B, \delout$ to be fixed constants. 
		Then, we compute upper bounds on $P^{(1) \rightarrow (2)}$, $P^{(1) \rightarrow (0)}$, $P^{(2) \rightarrow (1)}$, $P^{(2) \rightarrow (0)}$. Plugging these upper bounds to Equation~\eqref{eq:gamma-bdc}, we get an upper bound on $\gamma$ which we denote by $\tilde{\gamma}$. \footnote{We do not compute the value of $\gamma$ exactly as it is too difficult to do parametrically.}
		Note that $\tilde{\gamma}$ depends only on $M_1, T, M_2, \beta_1,$ and $p$. Then we choose $\delin$ to be larger than $\tilde{\gamma}$, and in particular we have $\gamma \leq \tilde{\gamma} < \delin$. Proposition~\ref{prop:decode} guarantees that if we choose a small enough $\epsout$, then our decoding algorithm will succeed with high probability. Thus, we only have to make sure that the rate that we get satisfies the statement in the theorem. We calculate the value of $\Rin$ using Proposition \ref{inner_code_lem} and then use it to calculate the overall rate. 
		
		We consider several regimes of $p$ and for each regime we choose suitable parameters.

		\paragraph{Case $p \geq 0.9$:}
		In this case we choose: 
		$$M_1 = 5.41, M_2 = 22.8, \beta_1 = 0.522,M_B = 10^{-5}, \delout = 2^{-20}  \quad \text{and } \delin = 0.01052 \;,$$ 
		and set $T=12$. From Proposition \ref{inner_code_lem} we get that, for our choice of parameters, the rate of  the inner code is $\Rin = 0.5229$. 
		The upper bounds for $P^{(1) \rightarrow (0)}, P^{(2) \rightarrow (1)}, P^{(2) \rightarrow (0)}$ are computed using Equations \eqref{full_deletion_bound_1}, \eqref{big_to_small_run_bound}, and \eqref{full_deletion_bound_2}. To upper bound $P^{(1)\rightarrow (2)}$, we use Equation~\eqref{small_to_big_run_bound} given in Lemma~\ref{lem:small_to_big_run_bound} with $q = 0.1$. Observe that as we assume $p\geq 0.9$ it follows that $q\geq 1-p$.
		
		One can plug in the upper bounds to Equation~\eqref{eq:gamma-bdc} and observe that $\tilde{\gamma} < \delin$.
		Proposition~\ref{prop:decode} guarantees that for a small enough $\epsout$ our decoding algorithm succeeds with high probability.
		To calculate the rate we use Equation~\eqref{Rate_BDC_eq}. For a large enough  $m$ (e.g. $m> 10^5$) we obtain
		\[
		\frac{0.5229 (1-p)}{8.27323 + 0.761(1-p) + (1-p)/m} \geq \frac{0.5229 (1-p)}{8.34933} > \frac{(1-p)}{16}\; .
		\]
		
		\paragraph{Case $0.57 < p < 0.9$:} 
		For this regime we use the parameters 
		$$M_1 = 5.59, M_2 = 23.5, \beta_1 = 0.53,M_B = 10^{-5}, \delout = 2^{-20}  \quad \text{and } \delin = 0.008013 \;,$$ 
		and set $T = 13$. We get that the rate of the inner code is $\Rin = 0.55224$. We first note that the calculations used to upper bound $P^{(1) \rightarrow (0)}, P^{(2) \rightarrow (1)}, P^{(2) \rightarrow (0)}$ were obtained by using Equations~\eqref{full_deletion_prob_1},~\eqref{big_to_small_run_probability} and~\eqref{full_deletion_prob_2} with $p = 0.9$. 
		This can be done since Equations ~\eqref{full_deletion_prob_1} and ~\eqref{full_deletion_prob_2} are clearly monotonically increasing in $p$ and we are considering smaller values of $p$. Also, observe that since $M_2/(1 - 0.9) = 235$ is an integer, then by Equation \eqref{big_to_small_run_bound_q} given in Lemma~\ref{lem:bounds-big-run-small-run}, for every $p \leq 0.9$,
		\[
		P^{(2)\rightarrow (1)} \leq \sum_{i=0}^{T} \binom{\frac{M_2}{1-0.9}}{i} (1-0.9)^i \cdot (0.9)^{\frac{M_2}{1-0.9} - i} \;,
		\]
		which is exactly what we get from Equation~\eqref{big_to_small_run_probability} with $p = 0.9$. Now, to upper bound $P^{(1) \rightarrow (2)}$ we use Equation~\eqref{small_to_big_run_bound} with $q = 1 - 0.57$, which is fine as $p > 0.57$ and thus $q>1-p$. As before, calculations show that $\tilde{\gamma} < \delin$. Hence for a small enough $\epsout$ our decoding algorithm succeeds with high probability by Proposition~\ref{prop:decode}.
		Plugging the parameters into Equation~\eqref{Rate_BDC_eq} and letting $m$ be large enough we get 
		\[
		\frac{0.55224 (1-p)}{8.48521 + 0.765 (1-p) + (1-p)/m} > \frac{0.55224 (1-p)}{8.81416}> \frac{1-p}{16}\; .
		\]
		
		\paragraph{Case $0 < p \leq 0.57$:} The parameters we choose for this regime are 
		$$M_1 = 5.59, M_2 = 20.21, \beta_1 = 0.53,M_B = 10^{-5}, \delout = 2^{-20}  \quad \text{and } \delin = 0.006147 \;,$$ 
		and set $T = 13$. Using Proposition~\ref{inner_code_lem}, we get that $\Rin = 0.577475$.
		As in the previous case, the upper bounds to $P^{(1) \rightarrow (0)}, P^{(2) \rightarrow (1)}, P^{(2) \rightarrow (0)}$ were obtained by using Equations~\eqref{full_deletion_prob_1},~\eqref{big_to_small_run_probability} and~\eqref{full_deletion_prob_2}, this time with $p = 0.57$ (observe that $M_2/(1-0.57)$ is an integer).
		In this case $13 = T \geq \ceil{M_1/(1-p)}$. For a random variable $Z$ distributed as $Z \sim \textup{Bin} \left( \ceil{ M_1/(1-p)}, 1-p\right)$, it holds that
		\[
				P^{(1)\rightarrow (2)} = \Pr[Z \geq T+1] = 0
		\]
		since bits can only be deleted by the BDC$_p$.

%
		One can simply verify that $\tilde{\gamma} < \delin$ and hence for a small enough $\epsout$ our decoding algorithm succeeds with high probability by Proposition~\ref{prop:decode}.
		Plugging the parameters into Equation~\eqref{Rate_BDC_eq} and letting $m$ be large enough we get that for the case $p \leq 0.57$, the rate of the construction is 
		\[
		\frac{0.57747 (1-p)}{7.71206 + 0.765 (1-p) + (1-p)/m} > \frac{0.57747 (1-p)}{8.47706}> \frac{1-p}{16}\; .
		\]
		This completes the proof of Theorem~\ref{BDC_theorem}	
	\end{proof}

	\section{Rates For Fixed Values of Deletion Probabilities}\label{sec:fixed}
	
	In Theorem \ref{BDC_theorem} we constructed codes of rate larger than $(1-p)/16$ for the BDC$_p$ that can be used for reliable communication. Note that even if $p \rightarrow 1$ our construction gives codes of positive rate. Now, we wish to fix $p$ (and thus leave the regime $p\rightarrow 1$) and instead of using the bounds given in Equations~\eqref{small_to_big_run_bound},~\eqref{full_deletion_bound_1},~\eqref{big_to_small_run_bound}  and~\eqref{full_deletion_bound_2}, we can use the exact direct calculations given in Equations~\eqref{small_to_big_run_probability},~\eqref{full_deletion_prob_1},~\eqref{big_to_small_run_probability} and~\eqref{full_deletion_prob_2}, respectively. Using the exact bounds we can improve, for any fixed value of $p$, the rate of the code compared to what we obtained in Theorem~\ref{BDC_theorem}. 
	The reason that we can improve the bound is that in the proof of Theorem~\ref{BDC_theorem} we looked for a relatively simple argument that should work for every value of $p$. When $p$ is fixed, we can use more direct calculations to get a better bound. For example, we can get significant improvement by using Equation~\eqref{full_deletion_prob_1} instead of Equation~\eqref{full_deletion_bound_1}. E.g., for $p=0.8$  there is a relatively large difference between $p^{M_1/(1-p)}$ and $e^{-M_1}$. E.g., for $M_1 = 5$ we have that $e^{-5} = 0.00673$ and $0.8^{5/(0.2)} = 0.00377$. Such savings allow us to choose smaller value of $M_1$ for the case $p=0.8$. Then, by reducing the value of $M_1$ we reduce also the values of $T$ and $M_2$ which eventually lead to an improved rate.

	The reason that we do not optimize the calculation using these equations for every $p$ is that the optimization involves complex expressions involving all our parameters and it is not clear how to optimize it and get a closed formula for the rate for arbitrary $p$.
	

	In \cite{drinea2007improved}, the authors gave constructions of probabilistic codes 
	for the binary deletion channel. They derived lower bounds on the capacity of the BDC$_p$ that are the best lower bounds as far as we know for fixed values of $p$.
	
	In Table \ref{table:2} we compare our results to the ones obtained in \cite{drinea2007improved}. One can see that our rates are smaller by approximately a factor of $2$. Yet, the construction presented in this paper is deterministic, efficient, and has a simpler analysis. 
	
	\begin{table}[h]
		\centering
		\begin{tabular}{ | c | c | c || c |} 
			\hline
			$p$ & ($\beta_1, N_1, T, N_2, \Rin, \delin $) & Final rate & \cite{drinea2007improved} \\	
			\hline
			$0.50$ & ($0.497$, $8$, $7$, $27$, $0.5456$, $0.00922$) & $0.050682$ & $0.10186$ \\
			\hline
			$0.55$ & ($0.519$, $9$, $8$, $34$, $0.5525$, $0.00825$) & $0.043005$ & $0.084323 $\\
			\hline
			$0.60$ & ($0.508$, $10$, $8$, $38$, $0.5184$, $0.01120$) & $0.035935$ & $0.069564 $\\
			\hline
			$0.65$ & ($0.519$, $13$, $9$, $49$, $0.5545$, $0.00810$) & $0.029926$ & $0.056858 $\\
			\hline
			$0.70$ & ($0.509$, $15$, $9$, $57$, $0.5267$, $0.01051$) & $0.024353$ & $0.045324 $\\
			\hline
			$0.75$ & ($0.524$, $20$, $10$, $75$, $0.5400$, $0.00910$) & $0.019420$ & $0.035984 $\\
			\hline
			$0.80$ & ($0.514$, $24$, $10$, $96$, $0.5289$, $0.01022$) & $0.014830$ & $0.027266 $\\
			\hline
			$0.85$ & ($0.526$, $34$, $11$, $138$, $0.5413$, $0.00895$) & $0.010701$ & $0.019380 $\\
			\hline
			$0.90$ & ($0.537$, $54$, $12$, $224$, $0.5534$, $0.00773$) & $0.006845$ & $0.012378 $\\
			\hline
			$0.95$ & ($0.53$, $108$, $12$, $452$, $0.5402$, $0.00893$) & $0.003305$ & $0.005741 $\\
			\hline
			$0.99$ & ($0.52$, $541$, $12$, $2280$, $0.5318$, $0.00985$) & $0.000641$ & - \\
			\hline
		\end{tabular}
		\caption{Rates for fixed values of $p$. $N_1$ and $N_2$ are the lengths of the inner codeword runs after the blow-up. I.e., $N_1 = \ceil{M_1/(1-p)}$ and $N_2 = \ceil{M_2/(1-p)}$.}
		\label{table:2}
	\end{table}
	
	\begin{figure}
		\centering
		\begin{tikzpicture}
		
		\begin{axis}[
		scaled y ticks = false,
		tick label style={/pgf/number format/fixed },
		axis lines = left,
		xlabel = $p$,
		ylabel = {$Rate$},
		xtick={0.5,0.6,0.7,0.8,0.9,1},
		ytick={0,0.01,0.02,0.03,0.04,0.05}, 
		legend pos=north east,
		ymajorgrids=true,
		grid style=dashed,
		]
		\addplot [
		domain=0.5:1, 
		samples=100, 
		color=red,
		]
		{(1-x)/15.71};
		\addlegendentry{$(1-p)/15.71$}
		\addplot[
		color=blue,
		mark=square,
		]
		coordinates {
			(0.5,0.050682)(0.55,0.043005)(0.6,0.035935)(0.65,0.029926)(0.7,0.024353)(0.75,0.019420)(0.8,0.014830)(0.85,0.010701)(0.9,0.006845)(0.95,0.003305)(0.99,0.0006418)
		};
		\addlegendentry{Rates for fixed $p$}
		\end{axis}
		\end{tikzpicture}
		\caption{Rates for fixed values of $p$.} \label{fig:fixed_rates_vs_general}
	\end{figure}
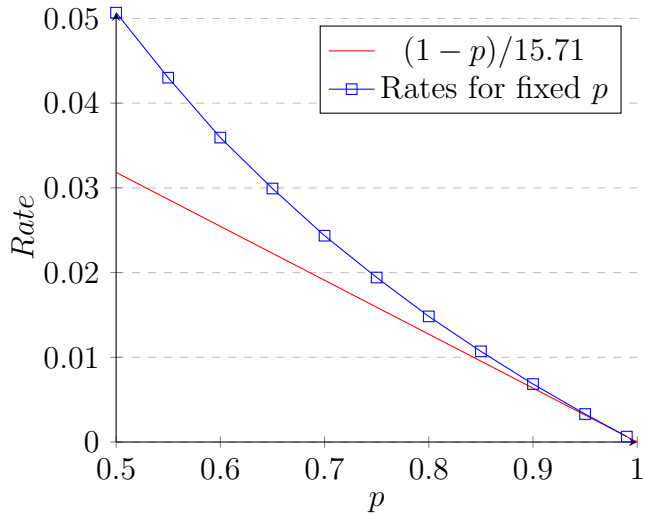

	Note that as $p$ tends to $1$ the rate that we achieve approaches $(1-p)/15.7$ as can be seen in figure \ref{fig:fixed_rates_vs_general}.

	\section{Poisson Repeat Channel}\label{sec:poisson}
	
		We first recall the definition of the PRC$_\lambda$.
	\begin{defi}
		Let $\lambda > 0$. The \emph{Poisson repeat channel with parameter $\lambda$ (PRC$_{\lambda}$)} replaces each transmitted bit randomly (and independently of other transmitted bits), with a discrete number of copies of that bit, distributed according to the Poisson distribution with parameter $\lambda$.
	\end{defi}

	This channel was first defined by Mitzenmacher and Drinea in \cite{mitzenmacher2006simple} who used it to prove a lower bound of $\left(1-p\right)/9$ on the rate of  the BDC. 	More recently, Cheraghchi \cite{cheraghchi2018capacity} gave an upper bound on its capacity and showed further connections to the BDC.

	Before proceeding, let us describe the connection between the PRC and the BDC discovered by Mitzenmacher and Drinea. What they observed is that a code for the PRC$_\lambda$ having rate $\mathcal R$, yields a code for the BDC$_p$ of rate $(1-p) \cdot \mathcal{R}/\lambda$. The reduction is via a probabilistic argument -- from each codeword in the code for the PRC$_{\lambda}$ we generate a codeword for the BDC$_p$ as follows: we replace each of the bits in the codeword by a discrete number of copies of those bits, distributed  according to the Poisson distribution with parameter $\lambda/(1-p)$. The intuition for the construction is that now, when we send the codeword through the BDC$_p$, the resulting word is distributed as if we had sent the original codeword  through the PRC$_\lambda$. 
	
	To the best of our knowledge, prior to this work there were no explicit  deterministic constructions of coding schemes for the PRC$_\lambda$.
In this section, we prove that the scheme that we constructed for the BDC  can also be used for PRC (with slightly different parameters). We note that one can also use the construction given in \cite{guruswami2017efficiently} to obtain a deterministic construction for the PRC, yet our construction yields better rates in this case as well.

We focus on the regime where $\lambda \leq 0.5$, as, in some sense, the PRC behaves like the BDC for small values of $\lambda$ -- intuitively, the smaller $\lambda$ is the more likely deletions are.


	We now describe the construction for this channel. Note that most of the details are identical to our construction for the BDC$_p$. Therefore, in order not to repeat the entire proof, we focus on the differences and leave the details to the reader.
	
	\subsection{Construction} \label{sec:poisson-construction}
	We  use the same inner and outer codes defined in Proposition \ref{inner_code_lem} and Theorem \ref{Haupler_thm}. For parameters $M_1 < M_2$ and $M_B$ our construction is as follows:
	
	\textbf{Encoding.}
	The only differences in the encoding procedure are the  length of the buffers and the blow-up of the runs:
	\begin{itemize}
		\item We place a buffer of $0$'s between every two inner codewords, where the buffers length is $ \ceil{M_B  m/\lambda}$.
		\item Every run of length $1$ is replaced with a run of length $\ceil{M_1/\lambda}$.
		\item Every run of length $2$ is replaced with a run of length $\ceil{M_2/\lambda}$.
	\end{itemize}
	\begin{remark}
		We must choose $M_2 > \lambda$ since otherwise all runs in the inner code will be replaced with a run of length $1$.
	\end{remark}
	\textbf{Decoding.}
	Since the inner and outer codes are the same we use the decoding algorithm given in Algorithm \ref{alg:decode}.\\
	
	\noindent\textbf{Rate.}
	Similar to the calculations yielding Equation~\eqref{Rate_BDC_eq}, the rate of this construction is	\begin{eqnarray}
	\mathcal{R} &=& \frac{\log\left(\left|\Sigma\right|^{\Rout n}\right) }{\beta_1 \ceil{ M_1/\lambda } nm + \beta_2 \ceil{ M_2 / \lambda } nm + \ceil{M_B m/\lambda} (n-1)} \nonumber \\
	&\geq & \frac{\Rin \Rout}{\beta_1 \ceil{ M_1/\lambda } + \beta_2 \ceil{ M_2 / \lambda }+ M_B/\lambda + 1/m}
\label{Rate_Poisson_eq}
 \\
	&\geq & \frac{\Rout \Rin \cdot \lambda}{\beta_1 M_1+ \beta_2 M_2 + \beta \lambda + M_B + \lambda / m} \;. \nonumber
	\end{eqnarray}
	As before, we can avoid the ceilings if we consider values of $\lambda$ such that $\ceil{ M_1 / \lambda }$, $\ceil{ M_2 / \lambda }$ and $\ceil{ M_B m / \lambda }$ are integers. In this case, the rate of the construction is given by
	\begin{equation} \label{Rate_Poisson_eq_no_ceil}
	\mathcal{R} \geq \frac{\Rin \Rout \cdot \lambda }{\beta_1 M_1+ \beta_2 M_2 + M_B} \;.
	\end{equation}
	
	\subsection{Correctness of Decoding Algorithm}
	Since we use the same inner and outer codes in our encoding and the same decoding algorithm, the analysis performed in Section~\ref{sec:analyze} can be repeated to this case as well with some minor modifications. We will briefly mention these modifications and leave the proofs to the reader.
	 
	We start by formally stating an analogous version of Proposition~\ref{prop:decode} to this setting.
	
	\begin{prop} \label{prop:decode-pois}
		Given $M_1, T, M_2, M_B, \beta_1, \delin, \epsin, \delout$ (as described in Section \ref{sec:poisson-construction}) let
		$Z_1 \sim \textup{Poisson}(\lambda \ceil{M_1/\lambda})$ and $Z_2 \sim \textup{Poisson}(\lambda \ceil{M_2/\lambda})$. 
		Denote 
		\begin{align*}
		&P^{(1) \rightarrow (2)} := \Pr[Z_1 \geq T + 1] \;, \\
		&P^{(1)\rightarrow (0)} := \Pr[Z_1 = 0] \;, \\ 
		&P^{(2) \rightarrow (1)} := \Pr[Z_2 \leq T] \;, \\
		&P^{(2)\rightarrow (0)} := \Pr[Z_2 = 0] \;,
		\end{align*}
		
		and define
		\begin{equation}\label{eq:gamma-prc}
		\gamma := \beta_1 \cdot P^{(1) \rightarrow (2)} + \beta_2 \cdot P^{(2) \rightarrow (1)} + \left(2\beta_1+ \beta_2 \right) P^{(1)\rightarrow (0)}+ 4 \beta_2 P^{(2)\rightarrow (0)} \;.		
		\end{equation}
		Let $x\in\Sigma^{\Rout n}$ be a message and let $y$ be the string obtained after encoding $x$ using our code and transmitting it through the PRC$_{\lambda}$. 
		If $\gamma < \delin$, then there exists $\epsilon_0=\epsilon_0(M_1, T, M_2, M_B,\beta_1, \delin, \delout)$ such that for every $\epsout < \epsilon_0$ it holds that Algorithm~\ref{alg:decode} returns $x$ with probability $1-\exp\left(-\Omega(n)\right)$.
	\end{prop} 
	
	Note that the only difference between this proposition and Proposition ~\ref{prop:decode} is in the definitions of $Z_1$ and $Z_2$.
	Recall that the proof of Proposition~\ref{prop:decode} heavily relies on Propositions~\ref{clm:dele_buffer}, \ref{clm:spuo_buffer}, and \ref{clm:wrong_inner_dec}. Therefore, to prove Proposition~\ref{prop:decode-pois}, one needs to formally state and prove analogous versions of Propositions~\ref{clm:dele_buffer}, \ref{clm:spuo_buffer}, and \ref{clm:wrong_inner_dec} in the setting of the PRC.
	
	We first observe that it is very simple to prove the analogous claims to Propositions~\ref{clm:dele_buffer} and~\ref{clm:spuo_buffer} by using Lemma~\ref{lem:chernoff-poisson} instead of Lemma~\ref{lem:chernoff} (since our random variables are now distributed according to the Poisson distribution).
	Hence we omit the details. We thus have that the probability of each error type is $\exp(-\Omega(m))$ per inner codeword. We focus on analyzing the case where we might output a wrong inner codeword in Step \ref{decode:inner} of Algorithm \ref{alg:decode} (i.e. the case analyzed in Proposition~ \ref{clm:wrong_inner_dec}).

	\subsubsection{Wrong Inner Decoding}
	Note that as we consider the same threshold decoding step for decoding the inner windows (i.e., Step~\ref{decode:blow-up:end} in Algorithm~\ref{alg:decode}) and the same inner code, the claims of Section~\ref{Inner_Dec_Analysis} apply here as well. 
	The difference from Section~\ref{Inner_Dec_Analysis} is in the computations of the probabilities $P^{(1)\rightarrow (2)}$ , $P^{(1)\rightarrow (0)}$, $P^{(2)\rightarrow (1)}$, $P^{(2)\rightarrow (0)}$. We focus on these computations as they play a significant role in computing the rate in Theorem~\ref{PRC_theorem}.
	
	
	Recall that in the encoding process, a run $r_j$ is replaced with a run of length $\ceil{M_1/\lambda}$ or $\ceil{M_2/\lambda}$ depending on $r_j$'s length. As in Section \ref{Inner_Dec_Analysis}, define $Z_j$ to be the random variable corresponding to the number of bits from this blown-up run that survived the transmission through the PRC$_{\lambda}$. According to Lemma~\ref{lem:poisson-sum}, $Z_j \sim \textup{Poisson}(\lambda \ceil{M_1/\lambda})$ if $\left|r_j\right| = 1$ and $Z_j \sim \textup{Poisson}(\lambda \ceil{M_2/\lambda})$ if $\left|r_j\right| = 2$.  
	Let $r_j'$ be exactly as defined in Definition \ref{def:r'_j}. As before, we study the probability that $r_j \neq r_j '$:
	\begin{enumerate}
		\item If $\left|r_j\right| = 1$ then there are two possible types of errors:
	\begin{enumerate}
		\item $\left|r_j '\right| = 2$: The probability for this to happen is $P^{(1)\rightarrow (2)} := \Pr[Z_j \geq T+ 1]$. We next give two estimates, one is an exact calculation and the other is an upper bound. For every $\lambda$ we have
		\begin{equation}
		\begin{split}
		P^{(1)\rightarrow (2)} &= \Pr[Z_j \geq T + 1] \\
		&= 1 - \Pr[Z_j \leq T] \\ 
		&= 1 - e^{-\lambda\ceil{\frac{M_1}{\lambda}}} \sum_{i=0}^{T} \frac{(\lambda\ceil{\frac{M_1}{\lambda}})^i}{i!} \;. \label{PRC_small_to_big_probability}
		\end{split}
		\end{equation} 
		Let $Y$ be a random variable distributed as $Y\sim \textup{Poisson} (M_1 + \lambda)$. We can upper bound $P^{(1) \rightarrow (2)}$ by 
		\begin{equation}\label{PRC_small_to_big_bound}
		\begin{split}
		P^{(1)\rightarrow (2)} = \Pr[Z_j \geq T + 1] &= 1 - \Pr[Z_j \leq T] \\ &\leq 1 - \Pr[Y \leq T] \\
		&= 1 - e^{-M_1 - \lambda} \sum_{i=0}^{T} \frac{(M_1 + \lambda)^i}{i!} \;.
		\end{split}
		\end{equation} 
		where the inequality follows from Lemma \ref{lem:poisson_mono} by noting that $ \lambda \ceil{M_1/\lambda} \leq M_1 + \lambda$.
		
		\item $\left|r_j '\right| = 0$: In this case, $r_j$ was completely deleted by the channel. The probability for this to happen is
		\begin{equation} \label{PRC_1run_deletion_pro}
		P^{(1)\rightarrow (0)} = \Pr[Z_j = 0] = e^{-\lambda\ceil{M_1/\lambda}} \leq e^{-M_1} \;.
		\end{equation}
		
	\end{enumerate}
	
	\item If $\left|r_j\right| = 2$ then one of the following cases hold:
	\begin{itemize}
		\item $\left|r_j'\right| = 1$: The probability for this to happen is $P^{(2)\rightarrow (1)} := \Pr[Z_j \leq T]$. As before, the exact probability calculation is
		\begin{equation}
		\label{PRC_big_to_small_probability}
		P^{(2)\rightarrow (1)} = \Pr[Z_j \leq T] = e^{-\lambda\ceil{\frac{M_2}{\lambda}}} \sum_{i=0}^{T} \frac{(\lambda\ceil{\frac{M_2}{\lambda}})^i}{i!} \;. 
		\end{equation}
		Let $Y$ be a random variable distributed as $Y \sim \textup{Poisson}(M_2)$ then it holds that
		\begin{equation}
		P^{(2)\rightarrow (1)} = \Pr[Z_j \leq T] \leq \Pr[Y \leq T] = e^{-M_2} \sum_{i=0}^{T} \frac{M_2^i}{i!}\;,	\label{PRC_big_to_small_bound}
		\end{equation}
		where the inequality follows from Lemma \ref{lem:poisson_mono} by noting that $M_2 \leq \lambda \ceil{M_2/ \lambda}$.	
		\item $\left|r_j '\right| = 0$. The probability for this to happen is
		\begin{equation} \label{PRC_2run_deletion_pro}
		P^{(2)\rightarrow (0)} = \Pr[Z_j = 0] = e^{-\lambda\ceil{M_2/\lambda}} \leq e^{-M_2}\;.
		\end{equation}
	\end{itemize}
	
	\end{enumerate}

By using these estimates and proceeding exactly as in the proof of Proposition~\ \ref{clm:wrong_inner_dec} one gets that the probability of error in this case as well is $\exp\left(-\Omega(m)\right)$. Combining everything together the proof of Proposition~\ref{prop:decode-pois} follows similarly to the proof of Proportion \ref{prop:decode} . In particular,  Algorithm~\ref{alg:decode} decodes correctly in this setting as well. 

\subsubsection{Proof of Theorem~\ref{PRC_theorem}}

As in the proof of Theorem~\ref{BDC_theorem}, we first compute an upper bound on $\gamma$ (recall its definition in Proposition~\ref{prop:decode-pois}) that holds for all $\lambda \leq 0.5$, then we compute the rate of the inner code by using Proposition~\ref{inner_code_lem} and finally we compute the rate of our code using Equation~\ref{Rate_Poisson_eq}. 

The parameters we use for our construction are 
\[
M_1 = 5.49, M_2 = 24.2, \beta_1 = 0.532, M_B = 10^{-5} \quad \textup{and } \delta_{\textup{out}} = 2^{-20} \;.
\] 
We pick $T = 13$ and set $\delin = 0.00954$.

First observe that for every $\lambda > 0$, we can upper bound $P^{(1)\rightarrow (0)}, P^{(2)\rightarrow (1)}, P^{(2)\rightarrow (0)}$ using Equations \eqref{PRC_1run_deletion_pro}, \eqref{PRC_big_to_small_bound}, and \eqref{PRC_2run_deletion_pro} respectively.
As we assume $\lambda \leq 0.5$, we can upper bound $P^{(1)\rightarrow (2)}$ using Equation \eqref{PRC_small_to_big_bound} with $\lambda = 0.5$ (due to monotonicity implied by Lemma~\ref{lem:poisson_mono}). 
Plugging these upper bounds to Equation~\eqref{eq:gamma-prc}, we get an upper bound on $\gamma$ which, as before, we denote by $\tilde{\gamma}$. Calculating, it is simple to verify that $\gamma\leq \tilde{\gamma} < \delin$. Therefore, for a small enough $\epsout$, our decoding algorithm succeeds with high probability.	
Applying Proposition \ref{inner_code_lem} we get an inner code of rate $\Rin = 0.53186$ and by letting $m$ be large enough, the rate of our concatenated code according to Equation \eqref{Rate_Poisson_eq} is
	\[
	\mathcal{R} = \frac{0.5318 \lambda}{8.58349 + 0.766 \lambda + \lambda/m} > \frac{\lambda}{17} \;.
	\]

	\section{Open Questions}
	The main open question is to further improve the construction presented in this paper and close the gap to (and even surpass) the lower bound of $(1-p)/9$ on the capacity of the BDC$_p$. 
	Alternatively, we can ask to come up with a deterministic construction for the PRC$_{\lambda}$ that gives better rates. By the reduction from the PRC to the BDC this will improve upon the constructions for the BDC.

	Even though the capacity of the BDC$_p$  scales proportionally with $1-p$ for $p\rightarrow 1$, it is an interesting open question to understand if there is a constant $1/9 \leq \mu \leq 0.4143$ such that the capacity of the channel is $(\mu \pm o(1))(1-p)$ where $o(1)$ is w.r.t. $n$, the block length of the code.

	Another interesting question is the maximal deletion fraction $\delta$, for which for every $\epsilon > 0$, there exists a code with rate bounded away from $0$ that can handle $\delta - \epsilon$ fraction of adversarial deletions. One can easily see that $\delta = 1/2$ is an upper bound (we simply delete all $0$'s or all $1$'s). Bukh et al. \cite{bukh2017improved} showed that  $\delta \geq \sqrt{2} - 1$. An interesting open question is whether this gap can be closed. I.e., are there codes that are capable of correcting $1/2 -\epsilon $ adversarial deletion for every $\epsilon > 0$ and that have rate bounded away from $0$.
	
	\section*{Acknowledgment}
	We wish to thank an anonymous reviewer for pointing out that our inner code is robust against the edit distance adversary and not just a restricted adversary considered in an earlier version of this paper, thus simplifying our arguments.
	
	\bibliographystyle{alpha}
	\bibliography{BDC_construction_improvement}
	
	\appendix

\end{document}